%% file: v2.tex
\title{Steady-State Cascade Operators and their Role in Linear Control, Estimation, and Model Reduction Problems}

\author{John~W.~Simpson-Porco, ~\IEEEmembership{Senior Member,~IEEE}, 
  Daniele Astolfi, 
  and  Giordano Scarciotti,~\IEEEmembership{Senior Member,~IEEE}
  \thanks{J.~W.~Simpson-Porco is with the Department of Electrical and Computer Engineering, University of Toronto, Toronto, ON M5S 3G4, Canada (jwsimpson@ece.utoronto.ca).} 
\thanks{D. Astolfi is with 
Universit\'e Claude Bernard Lyon 1, CNRS, LAGEPP UMR 5007, 43 boulevard du 11 novembre 1918, F-69100, Villeurbanne, France 
 ({\small  
daniele.astolfi@univ-lyon1.fr}).}
  \thanks{G. Scarciotti is with the Department of Electrical and Electronic Engineering,
Imperial College London, London SW7 2AZ, U.K. ({\small g.scarciotti@ic.ac.uk}).} 
\thanks{Research supported by NSERC Discovery RGPIN-2024-05523 and ANR  via ALLIGATOR project (ANR-22-CE48-0009-01).}
}

\documentclass[journal,twoside,web,9pt]{ieeecolor}

\usepackage{generic}
\usepackage{graphicx}
\usepackage{epstopdf}
\usepackage{graphics,psfrag,theorem}
\usepackage{amsfonts,latexsym}
\usepackage[noadjust]{cite}

\usepackage[utf8]{inputenc}
\usepackage{tabularx,ragged2e,booktabs}

\usepackage{soul}

\usepackage{mathrsfs}

\usepackage{amssymb,amsmath,mathtools}

\usepackage{enumerate}
\usepackage{float}

\graphicspath{{Figures/}}

\usepackage[colorlinks,hidelinks]{hyperref}
\usepackage{bigints}
\usepackage{varwidth}

\usepackage{tabularx,ragged2e,booktabs,caption}
\usepackage[noadjust]{cite}

\newcommand{\overbar}[1]{\mkern 1.5mu\overline{\mkern-1.5mu#1\mkern-1.5mu}\mkern 1.5mu}

\def\RR{{\mathbb R}}    
\def\CC{{\mathbb C}}    

\newcommand{\bA}{\bf{A}}
\newcommand{\bB}{\bf{B}}
\newcommand{\bC}{\bf{C}}
\newcommand{\bD}{\bf{D}}
\newcommand{\bF}{\bf{F}}
\newcommand{\bG}{\bf{G}}
\newcommand{\bH}{\bf{H}}

\newcommand{\bX}{\bf{X}}

\DeclareMathOperator{\rank}{rank}

\newcommand{\intersection}{\ensuremath{\operatorname{\cap}}}

\newtheorem{theorem}{Theorem}
\newtheorem{proposition}{Proposition}

\newtheorem{remark}{Remark}

\newtheorem{corollary}{Corollary}

\newcommand{\real}{\mathbb{R}}
\newcommand{\eig}{\mathrm{eig}}
\newcommand{\complex}{\mathbb{C}}
\newcommand{\integer}{\mathbb{Z}}
\renewcommand{\j}{\mathrm{j}}
\newcommand{\setdef}[2]{\{#1 \;|\; #2\}}
\newcommand{\map}[3]{#1: #2 \rightarrow #3}

\newcommand{\tb}{\color{black}}

\newcommand{\rom}[1]{
#1{}_{\rm r}}

\newcommand{\define}{\coloneqq}
\newcommand\oprocendsymbol{\hbox{$\square$}}
\newcommand\oprocend{\relax\ifmmode\else\unskip\hfill\fi\oprocendsymbol}

\DeclareFontFamily{U}{BOONDOX-calo}{\skewchar\font=45 }
\DeclareFontShape{U}{BOONDOX-calo}{m}{n}{<-> s*[1.05] BOONDOX-r-calo}{}
\DeclareFontShape{U}{BOONDOX-calo}{b}{n}{<-> s*[1.05] BOONDOX-b-calo}{}
\DeclareMathAlphabet{\mathcalboondox}{U}{BOONDOX-calo}{m}{n}
\SetMathAlphabet{\mathcalboondox}{bold}{U}{BOONDOX-calo}{b}{n}
\DeclareMathAlphabet{\mathbcalboondox}{U}{BOONDOX-calo}{b}{n}

\DeclareMathAlphabet{\mathdutchcal}{U}{dutchcal}{m}{n}
\SetMathAlphabet{\mathdutchcal}{bold}{U}{dutchcal}{b}{n}
\DeclareMathAlphabet{\mathdutchbcal}{U}{dutchcal}{b}{n}

\newcommand{\Syl}{\mathdutchcal{S}}
\newcommand{\Syld}{\mathdutchcal{S}_{\rm d}}
\newcommand{\Sylp}{\mathdutchcal{S}_{\rm p}}

\renewcommand{\L}{\mathdutchcal{C}}
\newcommand{\Lp}{\mathdutchcal{C}_{\rm p}}
\newcommand{\Ld}{\mathdutchcal{C}_{\rm d}}

\newcommand{\Hp}{\mathdutchcal{H}_{\rm p}}
\newcommand{\Hd}{\mathdutchcal{H}_{\rm d}}

\newcommand{\Ros}{\mathcalboondox{R}}

\def\sylp{\mathcalboondox{s}_{\rm p}}
\def\syld{\mathcalboondox{s}_{\rm d}}
 \def\lp{\mathdutchcal{c}_{\rm p}}
 \def\ld{\mathdutchcal{c}_{\rm d}}

\newcommand{\Mom}{\mathcalboondox{M}}

\usepackage{tikz}
\usetikzlibrary{arrows}
\usetikzlibrary{decorations.pathmorphing}
\usetikzlibrary{decorations.markings}
\usetikzlibrary{snakes}
\usetikzlibrary{matrix}
\usetikzlibrary{calc}
\usetikzlibrary{patterns}
\usetikzlibrary{positioning}
\usetikzlibrary{shapes}
\usetikzlibrary{shapes.symbols}
\usetikzlibrary{shapes.geometric}
\usetikzlibrary{shadows.blur}
\usetikzlibrary{shapes.arrows}
\usetikzlibrary{shapes.multipart}
\usetikzlibrary{shapes.callouts}
\usetikzlibrary{shapes.misc}
\usetikzlibrary{arrows}
\usetikzlibrary{positioning}
\usepackage{makecell}

\tikzstyle{every node}=[font=\small]
\tikzstyle{every path}=[line width=0.8pt,line cap=round,line join=round]

\begin{document}

\maketitle

\begin{abstract}
    Certain linear matrix operators arise naturally in systems analysis and design problems involving cascade interconnections of linear time-invariant systems, including problems of stabilization, estimation, and model order reduction. We conduct here a comprehensive study of these operators and their relevant system-theoretic properties. The general theory is leveraged to delineate both known and new design methodologies for {\tb control and observation of cascades, and to characterize structural properties of reduced models}. Several entirely new designs arise from this systematic categorization, including new recursive and low-gain design frameworks for observation of cascaded systems. The benefits of the results beyond the linear time-invariant setting are demonstrated through preliminary extensions for nonlinear systems, with an outlook towards the development of a similarly comprehensive nonlinear theory.
\end{abstract}

\begin{IEEEkeywords}
Sylvester equations, recursive design, forwarding, backstepping, output regulation, observer design, model order reduction, moment matching, tuning regulator
\end{IEEEkeywords}

\maketitle

\section{Introduction}
\label{Sec:Intro}

Equations and inequalities involving matrix variables arise frequently in systems analysis and design problems. A prominent example is the \emph{Sylvester equation} \cite[Chp. 6]{ACA:05}
\begin{equation}\label{Eq:Sylvester}
\Syl(\bX) \define \bA \bX - \bX\bB = \bC\bD,
\end{equation}
a linear equation in the unknown matrix $\bX$, with 
$(\bA,\bB,\bC,\bD)$ 
given matrix data of appropriate dimensions. As is well known, when $\bA$ and $\bB$ have disjoint spectra, the linear operator $\Syl$ is invertible, and $\bX = \Syl^{-1}(\bC\bD)$ is the unique solution to \eqref{Eq:Sylvester}. An incomplete list of the system-theoretic applications of \eqref{Eq:Sylvester} include linear output regulation \cite{BAF:77,HLT-AS-MH:01}, observer design \cite{DGL:64}, eigenstructure assignment \cite{BS-SPB:88}, model order reduction \cite{KAG-AV-PVD:04a, ACA:05} and disturbance decoupling \cite{VLS:94}. 

{\tb Sylvester equations however are often \emph{intermediate} constructions in problems involving \emph{cascaded interconnections}, including model reduction, feedback stabilization, and estimator design. In these problems, the objects of prime importance are often secondary linear matrix operators $\L$ derived from $\Syl$, which take the form
\begin{equation}\label{Eq:IntroSSC}
\bD \mapsto \L(\bD) \define  \bF\Syl^{-1}(\bC\bD)\bG + \bD \bH
\end{equation}
with $(\bF,\bG, \bH)$ given matrix data of appropriate dimensions. For reasons that will become apparent, we call operators of the form \eqref{Eq:IntroSSC} \emph{steady-state cascade (SSC) operators}. In some design frameworks, the system-theoretic properties of operator's \emph{values} $\L(\bD)$ are the point of focus. For example, a controllability property of $\L(\bD)$ was leveraged in \cite{DA-LP-LM:22b} as part of a recursive stabilizer design procedure, while in observer design procedures following \cite{DGL:64}, $\L(\bD)$ ends up specifying the required control input matrix for the observer. Similarly, in the model reduction context \cite{GS-AA:24}, the values $\L(\bD)$ are related to the so-called \emph{moments} (roughly, frequency response samples) of the linear system, which should be matched by a lower-order reduced model. In contrast, when other design pathways are taken, the properties of $\L$ \emph{as a linear operator} appear to be more critical. For instance, surjectivity of $\L$ was the key property enabling a static gain design procedure in \cite{LC-JWSP:23d}. 
%
%
Despite their frequent appearance in various analysis and design problems, no systematic theoretical study of SSC operators is available in the literature, limiting both our fundamental understanding of these operators and the development of design methodologies based on them.
}


{\tb Another motivation for the current work comes from the desire to develop new recursive stabilizer and estimator design procedures for nonlinear systems. As will become clear in Section \ref{Sec:Applications}, SSC operators \eqref{Eq:IntroSSC} naturally appear when recursive design is pursued in the LTI setting. In the nonlinear setting, the Sylvester equation \eqref{Eq:Sylvester} generalizes into a partial differential equation in an unknown function $\mathbf{x}(\cdot)$, and the SSC operator \eqref{Eq:IntroSSC} generalizes accordingly. While little is known regarding the global properties of such nonlinear SSC operators, by linearization, their local properties are determined by the properties of associated linear SSC operators (e.g., \cite[Lemma 1]{DA-LP-LM:22b}). It follows then that a thorough understanding of SSC operators in the LTI context has some immediate implications for local nonlinear design. Moreover though, as we will see in Sections \ref{Sec:Applications} and \ref{Sec:Nonlinear}, a clear delineation of LTI design procedures based on SSC operators immediately inspires new nonlinear recursive design problems and procedures.
}



\smallskip


\emph{Contributions:} 
{\tb Our preliminary work \cite{DA-JWSP-GS:23j} surveyed known applications of Sylvester and invariance equations in linear and nonlinear systems theory. In contrast, the present paper identifies SSC operators as the key unexplored objects of interest, and provides a comprehensive treatment of these operators, their properties, and their application in analysis and design problems.} Specifically, the paper contains three main contributions. First, in Section \ref{Sec:SSC} we provide a unifying theoretical study of SSC operators arising from cascade interconnections of linear time-invariant (LTI) state-space systems. There are two natural SSC operators one can define, depending on the order of the interconnection, and they are treated in parallel. Connections between these operators and the moments of the underlying systems are established. {\tb The key novel technical results are in Section \ref{Sec:SystemTheoreticSSC}}, wherein we establish relationships between ``non-resonance conditions'' on the plant data, injectivity and surjectivity of the SSC operators, and controllability and observability properties of their images. 

Second, in Section \ref{Sec:Applications} we apply the general theory of Section \ref{Sec:SystemTheoreticSSC} to the problems of cascade stabilization, cascade estimation, and model order reduction. For each case, we describe how different properties of the SSC operators may be leveraged to obtain different designs under varying assumptions. The results provide a library of design approaches, and highlights parallels between the distinct problems. Some of the specific design procedures outlined are known, while others are new. Of note, our treatment of disturbance estimation in this framework lead immediately to a novel low-gain design that is dual to the so-called tuning regulator framework \cite{EJD:76, LC-JWSP:23d} in the area of linear output regulation. 

The results of Sections \ref{Sec:SSC} and \ref{Sec:Applications} {\tb provide a map for nonlinear extensions; some of the theoretical and design results for LTI systems have known nonlinear counterparts, while some do not.} As a third contribution, in Section \ref{Sec:Nonlinear} {\tb we define nonlinear versions of SSC operators}, comment on the known nonlinear counterparts to our LTI design procedures, and identify unexplored extensions of the LTI design results for nonlinear systems. While this is primarily intended as an agenda for future research, as a concrete illustration of how our catalog of LTI designs may inspire new nonlinear design techniques, we present a novel recursive observer design for a cascade in which a nonlinear signal generator drives an LTI plant. {\tb Section \ref{Sec:Conclusions} concludes and lists further open directions.}

\medskip

\subsubsection*{Notation}
We denote with $\RR$, resp. $\CC$ the set of real, resp. complex, numbers. The symbol $\complex_{\geq 0}$ denotes the set of
complex numbers with non-negative real part, with $\complex_{>0}, \complex_{\leq 0}$, and so forth having the obvious meanings.  The set of eigenvalues of a matrix $A$ is denoted by $\eig(A)$. Given a positive integer $n$, $I_n$ denotes the $n \times n$ identity matrix. If $A \in \complex^{n \times m}$ and $\mathcal{V} \subset \complex^{m}$ is a subspace, then $A\mathcal{V} = \setdef{Av}{v \in \mathcal{V}}$. Given vectors (or matrices with the same number of columns) $x_1,\ldots,x_N$, $\mathrm{col}(x_1,\ldots,x_n)$ denotes their vertical concatenation. The symbol ``$\define$'' indicates a definition. {\tb Let $\mathcal{E}(\real;\real^{n_1\times n_2})$ denote the set of all continuous maps of $\epsilon \mapsto E(\epsilon) \in \real^{n_1 \times n_2}$ such that every element of $E$ is $O(\epsilon)$ as $\epsilon \to 0^+$, i.e., such that $\lim_{\epsilon \to 0^{+}}\Vert E(\epsilon)\Vert / \epsilon$ is finite.}

\section{The Steady-State Cascade Operators}
\label{Sec:SSC}



This section establishes technical results relating to two Sylvester equations and two associated linear operators, which we term the steady-state cascade (SSC) operators. Applications of these results to problems of stabilization, estimation, and model order reduction will be discussed in Section \ref{Sec:Applications}.

\subsection{Definitions and Interpretation}
\label{Sec:SSC_DI}

With $n, \nu, m, p$ positive integers, let $\Sigma = (A,B,C,D)$ with $A \in \complex^{n \times n}$, $B \in \complex^{n \times m}$, $C \in \complex^{p \times n}$, and $D \in \complex^{p \times m}$, and let $\Sigma^{\prime} = (F,G,H,J)$ with $F \in \complex^{\nu \times \nu}$, $G \in \complex^{\nu \times p}$, $H \in \complex^{m \times \nu}$, and $J \in \complex^{m \times p}$. When convenient, we will interpret these quadruples as defining LTI systems
\begin{equation}\label{Eq:SigmaSigma}
\Sigma:\,\,\begin{cases}\,\begin{aligned}
\dot{x} &= Ax + Bu\\
y &= Cx + Du
\end{aligned}\end{cases} \qquad \Sigma^{\prime}:\,\,\begin{cases}\,\begin{aligned}
\dot{\eta} &= F\eta + Gv\\
z &= H\eta + Jv
\end{aligned}\end{cases}
\end{equation}
with states $x \in \complex^{n}$ and $\eta \in \complex^{\nu}$ and associated transfer matrices
\begin{subequations}\label{eq:TF_definition}
    \begin{align}\label{eq:TF_definition_Sigma}
        \hat\Sigma(s) &= C(sI_n-A)^{-1}B+D,
        \\
        \hat\Sigma^\prime(s)& = H(sI_\nu-F)^{-1}G+J. \label{eq:TF_definition_SigmaPrime}
    \end{align}
\end{subequations}
We associate with $\Sigma$ and $\Sigma^{\prime}$, for lack of better terminology, \emph{primal and dual Sylvester operators}
\begin{subequations}
\begin{align}
&\map{\Sylp}{\complex^{n\times \nu}}{\complex^{n\times \nu}}, \quad \Sylp(\Pi) = \Pi F - A\Pi,\\
&\map{\Syld}{\complex^{\nu \times n}}{\complex^{\nu \times n}}, \quad \Syld(M) = MA - FM
\end{align}
\end{subequations}
If $\eig(A) \cap \eig(F) = \emptyset$, both $\Sylp$ and $\Syld$ are invertible linear operators, and the associated \emph{Sylvester equations}
\begin{subequations}\label{Eq:GenSyl}
\begin{align}
\label{Eq:GenSyl1}
\Sylp(\Pi) &= \Pi F - A\Pi = BH,\\
\label{Eq:GenSyl2}
\Syld(M) &= MA - FM = GC,
\end{align}
\end{subequations}
have unique solutions $\Pi = \Sylp^{-1}(BH)$ and $M = \Syld^{-1}(GC)$, respectively. To this point, $\Sigma$ and $\Sigma^{\prime}$ have been treated symmetrically. We elect now to think of $(A,B,C,D)$ and $F$ as fixed data, and we interpret the above solutions as \emph{linear} functions of $H$ and $G$, respectively. Based on this, we call the linear operators
\[
\map{\Lp}{\complex^{m \times \nu}}{\complex^{p \times \nu}}, \qquad \map{\Ld}{\complex^{\nu \times p}}{\complex^{\nu \times m}}
\]
defined by
\begin{subequations}\label{Eq:GenL}
\begin{align}
\label{Eq:GenL1}
\Lp(H) &\define C\Pi + DH = C\Sylp^{-1}(BH) + DH\\
\label{Eq:GenL2}
\Ld(G) &\define -MB + GD = -\Syld^{-1}(GC)B + GD
\end{align}
\end{subequations}
the \textit{steady-state cascade operators}. 

To provide some insight into these constructions, consider first the cascade interconnection $\Sigma^{\prime} \to \Sigma$ in which $\Sigma^{\prime}$ drives $\Sigma$ with $u = z$, and the output $y$ is observed, as shown in Figure \ref{Fig:Cascade1}. 
\begin{figure}[ht!]
\begin{center}
\input fig/fig_cascade.tex
\caption{The cascade interconnection $\Sigma^{\prime} \to \Sigma$.}
\label{Fig:Cascade1}
\end{center}
\end{figure}

The equations describing {\tb the} interconnection are
\begin{equation}\label{Eq:CascadePrimal}
\begin{aligned}
\begin{bmatrix}\dot{x} \\ \dot{\eta}\end{bmatrix} &= \begin{bmatrix}A & BH\\
0 & F
\end{bmatrix}\begin{bmatrix}x \\ \eta\end{bmatrix} + \begin{bmatrix}BJ \\ G\end{bmatrix}v\\
y &= \begin{bmatrix}C & DH\end{bmatrix}\begin{bmatrix}x \\ \eta\end{bmatrix} + \begin{bmatrix}DJ\end{bmatrix}v.
\end{aligned}
\end{equation}
 For \eqref{Eq:CascadePrimal} when $v \equiv 0$, the matrix $\Pi = \Sylp^{-1}(BH)$ defines an \emph{invariant subspace} $\setdef{(x,\eta)}{x = \Pi\eta}$ for the dynamics \eqref{Eq:CascadePrimal}. Motivated by this, by defining the error variable $\xi \define x - \Pi\eta$, {\tb the dynamics \eqref{Eq:CascadePrimal} become}
\begin{equation}\label{Eq:CascadePrimalTransformed}
\begin{aligned}
\begin{bmatrix}\dot{\xi} \\ \dot{\eta}\end{bmatrix} &= \begin{bmatrix}A & 0\\
0 & F
\end{bmatrix}\begin{bmatrix}\xi \\ \eta\end{bmatrix} + \begin{bmatrix}-\Pi G + BJ \\ G\end{bmatrix}v\\
y &= \begin{bmatrix}C & \Lp(H)\end{bmatrix}\begin{bmatrix}\xi \\ \eta\end{bmatrix} + \begin{bmatrix}DJ\end{bmatrix}v.
\end{aligned}
\end{equation}
When $\xi \equiv 0$ and $v \equiv 0$, we obtain the unforced dynamics on the invariant subspace, which are now simply described by
\begin{equation}\label{Eq:CascadePrimalObservation}
\dot{\eta} = F\eta, \qquad y = \Lp(H)\eta.
\end{equation}
{\tb The matrix $\Lp(H)$ is the observation matrix for the autonomous dynamics restricted to the invariant subspace.} Observe that if $A$ is Hurwitz, then trajectories of \eqref{Eq:CascadePrimalTransformed} converge to this invariant subspace, and if a steady-state exists (in the sense of, e.g., \cite{AI-CIB:08}), then $\Lp(H)$ describes the steady-state observation matrix relating $y$ to the state $\eta$ of the driving system. This scenario is the motivation behind our nomenclature \emph{steady-state} cascade operator.

Consider now the reverse cascaded system $\Sigma \to \Sigma^{\prime}$ in which $\Sigma$ drives $\Sigma^{\prime}$ with $v= y$, and the input $u$ is to be manipulated, as shown in Figure \ref{Fig:Cascade2}.
\begin{figure}[ht!]
\begin{center}
\input fig/fig_cascade2.tex
\caption{The cascade interconnection $\Sigma \to \Sigma^{\prime}$.}
\label{Fig:Cascade2}
\end{center}
\end{figure}

The equations describing the interconnection are
\begin{equation}\label{Eq:CascadeDual}
\begin{aligned}
\begin{bmatrix}\dot{x} \\ \dot{\eta}\end{bmatrix} = \begin{bmatrix}A & 0\\
GC & F
\end{bmatrix}\begin{bmatrix}x \\ \eta\end{bmatrix} + \begin{bmatrix}B \\ GD\end{bmatrix}u,
\end{aligned}
\end{equation}
where we omit the output $z$ as it will not be of interest. With $u \equiv 0$, the matrix 
$M = \Syld^{-1}(GC)$
defines an invariant subspace 
$\setdef{(x,\eta)}{\eta = Mx}$ 
{\tb for the joint dynamics \eqref{Eq:CascadeDual}}. Defining the error variable $\zeta \define \eta - Mx$, {\tb \eqref{Eq:CascadeDual} becomes}
\begin{equation}\label{Eq:CascadeDualTransformed}
\begin{aligned}
\begin{bmatrix}\dot{x} \\ \dot{\zeta}\end{bmatrix} &= \begin{bmatrix}A & 0\\
0 & F
\end{bmatrix}\begin{bmatrix}x \\ \zeta\end{bmatrix} + \begin{bmatrix}B \\ \Ld(G)\end{bmatrix}u\\
\end{aligned}
\end{equation}
{\tb The matrix $\Ld(G)$ is precisely the input matrix for the dynamics of the deviation variable $\zeta$, and thus describes how the control $u$ impacts the deviation from the invariant subspace. }

The derivations above suggest that $\Lp$ is most naturally associated with an observer design problem, while $\Ld$ is most naturally associated with a controller design problem; {\tb indeed, this will be the case. We will however} see that \textemdash{} under mild additional assumptions \textemdash{} $\Lp$ is nonetheless useful for controller design (Section \ref{Sec:LowGainStabLp}), and $\Ld$ is also useful for observer design (Section \ref{Sec:LowGainEstimatorLd}). 

{\tb
\begin{remark}[\bf Literature on Sylvester Equations]
The equation \eqref{Eq:Sylvester} has a long history of theoretical and numerical study. Among many references, see \cite{VK:74}, \cite[Chapter 6]{ACA:05}, \cite{MK-VM-PP:00}, 
for overviews of solvability properties and solution representations, \cite{EDS-SPB:81} for useful rank and controllability properties of solutions, and \cite{WA-FR-AS:94} for infinite-dimensional operator extensions. Computational methods based on the Schur decomposition for solving \eqref{Eq:Sylvester} are quite mature, and overviews of standard numerical methods can be found in \cite[Chapter 6]{ACA:05}, \cite{VS:16}. \hfill \oprocend
\end{remark}
}

\subsection{System-Theoretic Properties of the SSC Operators}
\label{Sec:SystemTheoreticSSC}

We now provide properties of the SSC operators that will be subsequently leveraged to develop different design pathways for a variety of system control problems. {\tb As motivation for the results}, consider again the cascade of Figure \ref{Fig:Cascade2}, leading to the transformed system \eqref{Eq:CascadeDualTransformed} and in particular to the simple deviation dynamics $\dot{\zeta} = F\zeta + \Ld(G)u$. For controller design purposes, we may now wonder when the pair $(F,\Ld(G))$ is stabilizable or controllable. Alternatively, we may wish to define a matrix $\overbar{G} \in \complex^{\nu \times m}$ such that $(F,\overbar{G})$ is stabilizable or controllable, and then ask whether there exists a (possibly, unique) matrix $G$ such that $\Ld(G) = \overbar{G}$. Analogous questions apply to the cascade of Figure \ref{Fig:Cascade1}.

%
%
The following results address these questions by providing a number of equivalent characterizations for the desired properties. As notation, for $\lambda \in \complex$ let
\begin{equation}\label{Eq:Rosenbrock}
\Ros_{\Sigma}(\lambda) \define \begin{bmatrix}A - \lambda I_{n} & B\\
C & D\end{bmatrix} \in \complex^{(n+p) \times (n+m)}
\end{equation}
be the \emph{Rosenbrock system matrix} associated with $\Sigma = (A,B,C,D)$.

\begin{theorem}[SSC Operators and Right-Invertibility]\label{Thm:SSC}
Suppose that $\eig(A) \cap \eig(F) = \emptyset$ and consider the operators $\Lp$ and $\Ld$ defined in \eqref{Eq:GenL}. The following statements are equivalent:
\begin{enumerate}[(i)]
\item \label{SSCNewItm:1} $\Ros_{\Sigma}(\lambda)$ has full row rank for all $\lambda \in \eig(F)$;
\item \label{SSCNewItm:2} For any pair $(P,Q)$, the system of equations
\begin{equation}\label{Eq:Francis}
\begin{aligned}
\Pi F  &= A\Pi + B\Psi + P\\
0 &= C\Pi + D\Psi + Q
\end{aligned}
\end{equation}
admits a solution $(\Pi,\Psi)$;
\item \label{SSCNewItm:3} For any pair $(\overbar{P},\overbar{Q})$ such that the system of equations
\begin{equation}\label{Eq:DualFrancis}
\begin{aligned}
MA &= FM + GC + \overbar{P}\\
0 &= -MB + GD + \overbar{Q}
\end{aligned}
\end{equation}
admits a solution $(M,G)$, the solution is unique;
\item \label{SSCNewItm:4} $\Lp$ is surjective; 
\item \label{SSCNewItm:5} $\Ld$ is injective.
\end{enumerate}
Moreover,
\begin{enumerate}[(a)]
\item  \label{SSCNewItm:6a} $(F,G)$ controllable and \eqref{SSCNewItm:1} $\Rightarrow$ $(F,\Ld(G))$ controllable,
\item  \label{SSCNewItm:6b} $(F,G)$ stabilizable and $\Ros_{\Sigma}(\lambda)$ full row rank {\tb for all} $\lambda \in \eig(F) \cap \complex_{\geq 0}$ $\Rightarrow$ $(F,\Ld(G))$ stabilizable,
\end{enumerate}
and the converses of \eqref{SSCNewItm:6a}, 
resp. \eqref{SSCNewItm:6b}, holds if $G^{\sf T}\ker(\lambda I_{\nu}-F^{\sf T}) = \complex^p$ for all $\lambda \in \eig(F)$, resp. for all $\lambda \in \eig(F) \cap \complex_{\geq 0}$.
\end{theorem}

A few comments are in order before proceeding to the proof. If $\hat{\Sigma}(s)$ denotes the transfer matrix from \eqref{eq:TF_definition_Sigma}, then
\[
\begin{bmatrix}-I_{n} & 0\\
C(\lambda I_{n}-A)^{-1} & I_p\end{bmatrix}\Ros_{\Sigma}(\lambda) = \begin{bmatrix}\lambda I_{n}-A & -B\\
0 & \hat{\Sigma}(\lambda)\end{bmatrix},
\]
since $\lambda \notin \mathrm{eig}(A)$ by assumption, from which it follows that \eqref{SSCNewItm:1} is equivalent to $\hat{\Sigma}(\lambda)$ having full row rank for all $\lambda \in \eig(F)$. Such a transfer matrix has rank $p$ for almost all $\lambda \in \complex$, and is called \emph{right invertible}, hence the theorem title. Item \eqref{SSCNewItm:2} is existence (but not uniqueness) of a solution to the traditional Francis regulator equations \cite{BAF:77}, while \eqref{SSCNewItm:3} is uniqueness (but not existence) of a solution to a natural dual set of equations. Items \eqref{SSCNewItm:4} and \eqref{SSCNewItm:5} provide corresponding statements regarding solvability of the linear operator equations $\Lp(H) = \overbar{H}$ and $\Ld(G) = \overbar{G}$, where $\overbar{H} \in \complex^{p \times \nu}$ and $\overbar{G} \in \complex^{\nu \times m}$. Regarding the final set of statements, if it is assumed at the outset that $(F,G)$ is controllable and that $G^{\sf T}\ker(\lambda I_{\nu}-F^{\sf T}) = {\tb \complex^p}$ for all $\lambda \in \eig(F)$, then controllability of $(F,\Ld(G))$ is \emph{equivalent} to statements \eqref{SSCNewItm:1}--\eqref{SSCNewItm:5}. 

{\tb While the equivalence \eqref{SSCNewItm:1} $\Longleftrightarrow$ \eqref{SSCNewItm:2} is classical, the equivalences to and between \eqref{SSCNewItm:3}, \eqref{SSCNewItm:4}, \eqref{SSCNewItm:5} are new contributions. A special version of the equivalence $(F,\Ld(G))$ controllable $\Leftrightarrow$ \eqref{SSCNewItm:1} $\Leftrightarrow$ \eqref{SSCNewItm:2} was presented in \cite[Proposition 2]{DA-LP-LM:22b}, but Theorem \ref{Thm:SSC} removes unnecessary assumptions on $(A,B)$ and $F$, and simplifies the proof; the statement in (b) and its converse are also new.}

\begin{remark}[\bf Cascade Controllability]
Controllability of LTI cascades is a classical problem, and was originally addressed in \cite{EJD-SW:75} via application of the Popov-Belevitch-Hautus (PBH) test. In particular, the cascaded system of Figure \ref{Fig:Cascade2} is controllable if and only if $(A,B)$ is controllable and for all $\lambda \in \eig(F)$ it holds that
\begin{equation}\label{Eq:DavisonPBH}
\mathrm{rank} \begin{bmatrix}
A - \lambda I_{n} & 0 & B\\
GC & F-\lambda I_{\nu} & GD
\end{bmatrix} = n+\nu.
\end{equation}
This if and only if condition depends on a mixture of data from the systems $\Sigma$ and $\Sigma^{\prime}$, and is difficult to generalize beyond the LTI case. On the other hand, the condition $(F,G)$ controllable along with any of \eqref{SSCNewItm:1}--\eqref{SSCNewItm:5} are together sufficient for \eqref{Eq:DavisonPBH}, and these slightly stronger formulations admit useful nonlinear extensions (Section \ref{Sec:Nonlinear})
and extensions to infinite-dimensional systems (see, e.g., \cite{VN:21}). 
\hfill \oprocend
\end{remark}

\begin{remark}[\bf Meaning of $\boldsymbol{G^{\sf T}\ker(\lambda I_{\nu}-F^{\sf T})} = \boldsymbol{\complex^p}$]
The condition $G^{\sf T}\ker(\lambda I_N-F^{\sf T}) = \complex^p$ for all $\lambda \in \eig(F)$ implies that $G$ has full column rank and that all eigenspaces of $F$ are at least $p$-dimensional. In fact, if all eigenspaces are exactly $p$-dimensional, then the condition implies that $(F,G)$ is controllable. This situation occurs, for instance, in linear output regulation design (see, e.g., \cite[Chapter 4]{AI:17}), wherein $(F,G) = (F_{\rm im} \otimes I_p,g_{\rm im} \otimes I_p)$ is a $p$-copy internal model for a single-input controllable pair $(F_{\rm im},g_{\rm im}) \in \real^{q \times q} \times \real^{q}$ with $q \in \integer_{\geq 1}$. \hfill \oprocend
\end{remark}

\begin{proof}
\eqref{SSCNewItm:1} $\Longleftrightarrow$ \eqref{SSCNewItm:2}: This result is classical (see, e.g. \cite[Theorem 9.6]{HLT-AS-MH:01}), and follows immediately by applying the surjectivity statement of Theorem \ref{Thm:HautusExtended} \eqref{Itm:HautusP} given in Appendix~\ref{Sec:Hautus} with $R_1 = \left[\begin{smallmatrix}-A & -B\\ C &D\end{smallmatrix}\right]$, $R_2 = \left[\begin{smallmatrix}I & 0\\ 0 &0\end{smallmatrix}\right]$, $X=\left[\begin{smallmatrix}\Pi\\ \Psi\end{smallmatrix}\right]$, $q_1(s) = 1$, and $q_2(s) = s$. 

\eqref{SSCNewItm:1} $\Longleftrightarrow$ \eqref{SSCNewItm:3}: Apply the injectivity statement of Theorem \ref{Thm:HautusExtended} \eqref{Itm:HautusD} 
given in Appendix~\ref{Sec:Hautus}
with the same selections of $R_1, R_2, q_1, q_2$ as in the above equivalence and $Y=\left[\begin{smallmatrix}M & G\end{smallmatrix}\right]$. 

\eqref{SSCNewItm:1} $\Longleftrightarrow$ \eqref{SSCNewItm:4}: By definition from \eqref{Eq:GenSyl1}, \eqref{Eq:GenL1}, $\Lp$ is surjective if for any $\overbar{H} \in \complex^{p \times \nu}$ there exists a solution $(\Pi,H)$ to $\Pi F = A\Pi + BH$ and $\overbar{H} = C\Pi + DH$ which we write together as
\begin{equation}\label{Eq:LpHautus}
\begin{bmatrix}0 \\ \overbar{H}\end{bmatrix} = \begin{bmatrix}
A & B\\
C & D
\end{bmatrix}\begin{bmatrix}
\Pi \\ H
\end{bmatrix} + \begin{bmatrix}
-I & 0\\
0 & 0
\end{bmatrix}\begin{bmatrix}
\Pi \\ H
\end{bmatrix}F.
\end{equation}
These equations are of course an instance of \eqref{SSCNewItm:2} with $(P,Q) = (0,\overbar{H})$, and thus \eqref{SSCNewItm:1} is certainly sufficient for solvability. For necessity, proceeding by contraposition, suppose that $\rank\,\Ros_{\Sigma}(\lambda) < n + p$ for some $\lambda \in \eig(F)$. Let $\varphi \in \complex^{\nu}$ be any non-zero vector such that $F\varphi = \lambda \varphi$, and let $w = \mathrm{col}(w_1,w_2) \in \complex^{n+p}$ be any non-zero vector such that $w^{\sf T}\Ros_{\Sigma}(\lambda) = 0$; the latter equations read as
\begin{equation}\label{Eq:RightInvertibleFails}
\begin{bmatrix}w_1^{\sf T} & w_2^{\sf T}\end{bmatrix}\begin{bmatrix}A-\lambda I_n\\ C\end{bmatrix} = 0, \qquad \begin{bmatrix}w_1^{\sf T} & w_2^{\sf T}\end{bmatrix}\begin{bmatrix}B \\ D\end{bmatrix}  = 0.
\end{equation}
Note that $w_2 \neq 0$; indeed, if $w_2 = 0$ then necessarily $w_1 \neq 0$, and the above relations imply that $w_1^{\sf T}A = \lambda w_1^{\sf T}$, so $\lambda \in \eig(A)$, which would contradict the assumption that $\eig(A) \cap \eig(F) = \emptyset$. Left and right-multiplying \eqref{Eq:LpHautus} by $w^{\sf T}$ and $\varphi$ and using the above relations we find that
\[
w_2^{\sf T}\overbar{H}\varphi = w^{\sf T}\Ros_{\Sigma}(\lambda)\begin{bmatrix}\Pi \\ H\end{bmatrix}\varphi = 0.
\]
Thus the equations \eqref{Eq:LpHautus} are insolvable for at least the particular choice $\overbar{H} = \mathrm{conj}(w_2)\varphi^* \neq 0$, since then $w_2^{\sf T}\overbar{H}\varphi = \|w\|_2^2\|\varphi\|_2^2 \neq 0$; this shows $\Lp$ is not surjective. 


\eqref{SSCNewItm:3} $\Longleftrightarrow$ \eqref{SSCNewItm:5}: For the forward direction, by definition $\Ld$ is injective if the only solution to $\Ld(G) = 0$ is $G = 0$, or equivalently, if the only solution to the equations
\begin{equation}\label{Eq:MG}
MA = FM + GC, \qquad 0 = -MB + GD
\end{equation}
is $(M,G) = (0,0)$. 
These equations are an instance of \eqref{SSCNewItm:3} with $(\overbar{P},\overbar{Q}) = (0,0)$ and $(0,0)$ is clearly a solution in this case, so injectivity follows. For the converse, if \eqref{SSCNewItm:3} fails then (by linearity) there exists a solution $(M,G) \neq (0,0)$ to \eqref{Eq:MG}. Moreover, this solution must satisfy $G \neq 0$, since if $G = 0$, the first of \eqref{Eq:MG} implies that $M = 0$ since $\eig(A) \cap \eig(F) = \emptyset$. In other words, we have found a non-zero $G$ such that $\Ld(G) = 0$, so $\Ld$ is not injective. 


Statements~(\ref{SSCNewItm:6a}) and its converse (and analogously statement (\ref{SSCNewItm:6b}) and its converse) will follow from the next statement we will prove: an eigenvalue $\lambda \in \eig(F)$ is controllable for the pair $(F,\Ld(G))$ if $\lambda$ is controllable for the pair $(F,G)$ and $\Ros_{\Sigma}(\lambda)$ has full row rank, and under the additional assumption that $G^{\sf T}\ker(\lambda I_{\nu}-F^{\sf T}) = \complex^p$, these two conditions are necessary. Let $G$ be given and select $\lambda \in \eig(F)$ with $\varphi \in \complex^{\nu}$ a left-eigenvector of $F$ associated with $\lambda$. With $M$ as defined in \eqref{Eq:GenSyl2}, left-multiplying \eqref{Eq:GenSyl2} by $\varphi^{\sf T}$ we obtain
\begin{equation}\label{Eq:TempRewritten1}
\varphi^{\sf T}MA - \lambda \varphi^{\sf T}M  =  \varphi^{\sf T}GC.
\end{equation}
Similarly, we have that $\varphi^{\sf T}\Ld(G) = \varphi^{\sf T}(-MB+GD)$. Grouping this equation with \eqref{Eq:TempRewritten1}, we have
\begin{equation}\label{Eq:TempRewritten2}
\begin{bmatrix}-\varphi^{\sf T}M & \varphi^{\sf T}G\end{bmatrix}\underbrace{\begin{bmatrix}A-\lambda I_n & B\\ C & D\end{bmatrix}}_{=\Ros_{\Sigma}(\lambda)} = \begin{bmatrix}0 & \varphi^{\sf T}\Ld(G)\end{bmatrix}.
\end{equation}
If $\lambda$ is controllable for $(F,G)$ and $\Ros_{\Sigma}(\lambda)$ has full row rank, then $\varphi^{\sf T}G \neq 0$ and the left-hand side of \eqref{Eq:TempRewritten2} cannot be zero, so we conclude that $\varphi^{\sf T}\Ld(G) \neq 0$; since $\varphi \in \ker(\lambda I_{\nu}-F^{\sf T})$ was arbitrary, controllability of $\lambda$ for the pair $(F,\Ld(G))$ follows from the eigenvector test. 

For the converse, assume now that $G^{\sf T}\ker(\lambda I_{\nu}-F^{\sf T}) = \complex^p$, and observe that since $G \in \complex^{\nu \times p}$, this condition implies that $G$ has full column rank and that $\dim \ker (\lambda I_{\nu}-F^{\sf T}) \geq p$. Again let $\varphi \in \complex^{\nu}$ be a left-eigenvector of $F$ associated with $\lambda$, leading by identical steps to \eqref{Eq:TempRewritten2}. We now proceed by contraposition. First, if $\lambda$ was uncontrollable for $(F,G)$, then we may take $\varphi$ in \eqref{Eq:TempRewritten2} such that $\varphi^{\sf T}G = 0$, and \eqref{Eq:TempRewritten2} then implies that $\varphi^{\sf T}\Ld(G) = 0$, which shows $\lambda$ is uncontrollable for $(F,\Ld(G))$. For the other case, suppose instead that $\lambda$ is controllable for $(F,G)$, but that $\Ros_{\Sigma}(\lambda)$ does not have full row rank. Then there must exist a non-zero vector $w = \mathrm{col}(w_1,w_2)$ such that $(w_1^{\sf T},w_2^{\sf T})\Ros_{\Sigma}(\lambda) = (0,0)$, which is written out previously in \eqref{Eq:RightInvertibleFails}. By arguments identical to those following \eqref{Eq:RightInvertibleFails}, it must be that $w_2 \neq 0$. Moreover, given $w_2$, since $\lambda \notin \eig(A)$, it follows from the first of \eqref{Eq:RightInvertibleFails} that $w_1$ is uniquely specified by $w_1^{\sf T} = w_2^{\sf T}C(\lambda I_n-A)^{-1}$. Consider now the feasibility of the linear equation $\varphi^{\sf T}G = w_2^{\sf T}$ in the unknown $\varphi \in \ker (\lambda I_{\nu}-F^{\sf T})$. Since $\mathrm{rank}(G) = p$ and $\dim \ker (\lambda I_{\nu}-F^{\sf T}) \geq p$, there must exist a choice of $\varphi \in \ker (\lambda I_{\nu}-F^{\sf T})$ such that $\varphi^{\sf T}G = w_2^{\sf T}$. Selecting this $\varphi$ in \eqref{Eq:TempRewritten2} yields $w_1^{\sf T}= \varphi^{\sf T}M$ from the first equation and, consequently, the second equation becomes $w_1^{\sf T} B + w_2^{\sf T}D=\varphi^{\sf T}\Ld(G)$. The second equation in \eqref{Eq:RightInvertibleFails} establishes that $\varphi^{\sf T}\Ld(G) = 0$, and thus $\lambda$ is uncontrollable for the pair $(F,\Ld(G))$; this completes the converse proof.
\end{proof}

\smallskip

The next result is dual in a very clear sense to Theorem \ref{Thm:SSC}; all proofs follow similar lines and are omitted. {\tb While the equivalence \eqref{SSC2NewItm:1} $\Leftrightarrow$ \eqref{SSC2NewItm:2} is certainly known (although perhaps not stated in this fashion) the remaining equivalences and statements are new.}

\begin{theorem}[SSC Operators and Left-Invertibility]\label{Thm:SSC2}
Suppose that $\eig(A) \cap \eig(F) = \emptyset$ and consider the operators $\Lp$ and $\Ld$ defined in \eqref{Eq:GenL}. The following statements are equivalent:
\begin{enumerate}[(i)]
\item \label{SSC2NewItm:1} $\Ros_{\Sigma}(\lambda)$ has full column rank for all $\lambda \in \eig(F)$;
\item \label{SSC2NewItm:2} For any pair $(P,Q)$ such that the system of equations \eqref{Eq:Francis} admits a solution $(\Pi,\Psi)$, the solution is unique;
\item \label{SSC2NewItm:3} For any pair $(\overbar{P},\overbar{Q})$ the system of equations \eqref{Eq:DualFrancis} admits a solution $(M,G)$;
\item \label{SSC2NewItm:4} $\Lp$ is injective; 
\item \label{SSC2NewItm:5} $\Ld$ is surjective;
\end{enumerate}
Moreover,
\begin{enumerate}[(a)]
\item  \label{SSC2NewItm:6a} $(F,H)$ observable and \eqref{SSC2NewItm:1} $\Rightarrow$ $(F,\Lp(H))$ observable,
\item  \label{SSC2NewItm:6b} $(F,H)$ detectable and $\Ros_{\Sigma}(\lambda)$ full column rank for all $\lambda \in \eig(F) \cap \complex_{\geq 0}$ $\Rightarrow$ $(F,\Lp(H))$ detectable,
\end{enumerate}
and the converses of \eqref{SSC2NewItm:6a} and \eqref{SSC2NewItm:6b} hold if $H\ker(\lambda I_{\nu}-F) = \complex^m$ for all $\lambda \in \eig(F)$ (resp. for all $\lambda \in \eig(F) \cap \complex_{\geq 0}$).
\end{theorem}


For completeness, we note that stronger statements still can be made in the case where the system $\Sigma$ is square (i.e., when $p = m$).

\begin{corollary}[SSC Operators and Invertibility]\label{Cor:SSC3}
Suppose that $\eig(A) \cap \eig(F) = \emptyset$ and consider the operators $\Lp$ and $\Ld$ defined in \eqref{Eq:GenL}. If $p = m$ then the following statements are equivalent:
\begin{enumerate}[(i)]
\item \label{SSC3Itm:1} $\det\,\Ros_{\Sigma}(\lambda) \neq 0$ for all $\lambda \in \eig(F)$;
\item \label{SSC3Itm:2} for any pair $(P,Q)$ (resp. $(\overbar{P},\overbar{Q})$) the system of equations \eqref{Eq:Francis} (resp. \eqref{Eq:DualFrancis}) admits a unique solution $(\Pi,\Psi)$ (resp. $(M,G)$);
\item \label{SSC3Itm:3} $\Lp$ and $\Ld$ are invertible.
\end{enumerate}
The controllability, observability, stabilizability, and detectability statements (and their converses) of Theorems \ref{Thm:SSC} and \ref{Thm:SSC2} continue to hold as stated.
\end{corollary}

\begin{remark}[\bf Effect of State Feedback and Output Injection]\label{Rem:StateFeedbackOutputInjection}
{\tb In problems of cascade stabilization (see Section \ref{Sec:Cascade-Stab}) and estimator design (see Section \ref{Sec:Cascade-Obsv}), it is convenient to enforce} the condition $\eig(A) \cap \eig(F) = \emptyset$ via state feedback and/or output injection applied to the system $\Sigma$, which leads to the transformations $A \rightarrow \mathcal{A} \define A + BK + LC$, $C \rightarrow \mathcal{C} \define C+DK$, and $B \rightarrow \mathcal{B} = B+LD$ for some matrices $K \in \complex^{m \times n}$ and $L\in \complex^{n \times p}$. The rank conditions \eqref{SSCNewItm:1} and \eqref{SSC2NewItm:2} however refer to the transmission zeros of $\Sigma$ which are invariant under these transformations. Put differently, {\tb the conditions of} Theorem \ref{Thm:SSC} \eqref{SSCNewItm:1}, \ref{Thm:SSC2} \eqref{SSC2NewItm:1}, or Corollary \ref{Cor:SSC3} \eqref{SSC3Itm:1} may be checked using either $(A,B,C,D)$ or $(\mathcal{A},\mathcal{B},\mathcal{C},D)$. \hfill \oprocend
\end{remark}

%
%
%
%
%


\subsection{Parameterization via Frequency Response and Moments}
\label{Sec:Moments}

At first glance, the SSC operators \eqref{Eq:GenL} would appear to depend densely on the data $(A,B,C,D)$ of the system $\Sigma$. Surprisingly however, these operators may be parameterized using only the so-called \emph{moments} of $\Sigma$ at the eigenvalues of the system matrix $F$ of $\Sigma^{\prime}$. This and related observations have been exploited for {\tb control design based on sampled frequency response data in
\cite{EJD:76, LP:16, VN:21, LC-JWSP:23d}} and for model order reduction in, e.g., \cite{GS-AA:17, GS-AA:24, MFS-GS-AYP-AP-NVDW:23}.

For convenience, our definition of moment differs slightly from what one typically finds in the literature. For $k \in \integer_{\geq 0}$ we define with 
$\Mom_k$
the \emph{$k$-th moment matrix of $\Sigma$ at $s_0 \in \complex \setminus \mathrm{eig}(A)$} as the complex matrix
\[
\Mom_{k}(s_0) = \frac{(-1)^k}{k!}\frac{\mathrm{d}^k}{\mathrm{d} s^k} \hat{\Sigma}(s)\Bigg|_{s = s_0} \in \complex^{p \times m},
\]
where $\hat{\Sigma}(s)$ is the transfer matrix \eqref{eq:TF_definition_Sigma}. It follows that
\[
\begin{aligned}
\Mom_{0}(s_0) &= C(s_0 I_{n}-A)^{-1}B + D = \hat \Sigma(s_0), &&\\
\Mom_{k}(s_0) &= C(s_0I_{n}-A)^{-(k+1)}B, && k \in \integer_{\geq 1}.
\end{aligned}
\]
{\tb Note that if $s_0 = \j\omega_0$, with $\j$ being the imaginary unit, the $0$-th moment matrix $\Mom_{0}(s_0)$ is simply the frequency response of the system $\Sigma$ evaluated at frequency $\omega = \omega_0$. As such, $0$-th moments essentially describe the steady-state response of a system to a prescribed type of input.}

{\tb The next result, Theorem \ref{Thm:Moments}, shows that SSC operators can be parameterized using moment matrices. To concisely state the result,} we require notation associated with a Jordan decomposition of $F$. Let $\{\lambda_1,\lambda_2\ldots,\lambda_{l}\}$ with $l \leq \nu$ denote the distinct eigenvalues of $F$, with associated algebraic multiplicities $\{m_1,\ldots,m_{l}\}$. Let $F = V\mathcal{J}V^{-1}$ denote a Jordan decomposition of $F$, where $\mathcal{J} = \mathrm{diag}(\mathcal{J}_1,\ldots,\mathcal{J}_{l})$ are the Jordan blocks with $\mathcal{J}_k \in \complex^{m_k \times m_k}$. We may always write $\mathcal{J}_k = \lambda_k I_{m_k} + N_k$, where $N_k \in \complex^{m_k \times m_k}$ is nilpotent. The matrix $V = \left[\begin{smallmatrix}V_1 & V_2 & \cdots & V_{l} & \end{smallmatrix}\right]$ is the transformation matrix with $V_k \in \complex^{\nu \times m_k}$ having full column rank. Partitioning $V^{-1}$ in accordance with $\mathcal{J}$, we may write
\[
V^{-1} = \mathrm{col}(
W_1^*, W_2^*,\ldots,W_{l}^*)
\]
for appropriate matrices $W_k \in \complex^{\nu \times m_k}$. Finally, for each $k \in \{1,\ldots,l\}$,  define the $\nu \times \nu$ matrices
\begin{equation}\label{Eq:Xkj}
X_{k,j} \define V_kN_k^{j}W_k^*, \qquad j \in \{0,\ldots,m_k-1\}.
\end{equation}

\begin{theorem}[\bf SSC Operators and Moments]\label{Thm:Moments}
 If $\eig(A) \cap \eig(F) = \emptyset$, then $\Lp$ and $\Ld$ are well-defined and
\begin{subequations}\label{Eq:LfFreq}
\begin{align}
\label{Eq:LfFreqa}
\Lp(H) &= \sum_{k=1}^{l}\nolimits\sum_{j=0}^{m_k-1}\nolimits(-1)^{j}\Mom_{j}(\lambda_k)HX_{k,j},\\
\label{Eq:LfFreqb}
\Ld(G) &= \sum_{k=1}^{l}\nolimits\sum_{j=0}^{m_k-1}\nolimits (-1)^{j}X_{k,j}G\Mom_j(\lambda_k).
\end{align}
\end{subequations}
\end{theorem}
\smallskip

The expression \eqref{Eq:LfFreqa} extends \cite[Theorem 2]{LC-JWSP:23d} to the case where the eigenvalues of $F$ are not simple, and can be shown to be equivalent to the expression in \cite[Equation (7)]{MFS-GS-AYP-AP-NVDW:23}. To our knowledge the expression \eqref{Eq:LfFreqb} is new. The main novelty of Theorem \ref{Thm:Moments} is therefore primarily in the technical proof, and in the symmetrical expressions \eqref{Eq:LfFreq} for $\Lp(H)$ and $\Ld(G)$ in terms of the {\tb moment matrices} of the system $\Sigma$. The main \emph{value} of these expressions is that moments can be obtained directly from input-output experiments on the plant; see, e.g., \cite{EJD:76,GS-AA:17,JM-GS:22} for details. {\tb This allows for a parameterization of the SSC operators based directly on experiments.}

\smallskip

\begin{proof}
We prove the first expression; the proof for the second is nearly identical and thus omitted. The proof combines ideas from \cite[Chp. 6]{ACA:05} and \cite{LC-JWSP:23d}. For $s \in \complex$, add $s\Pi$ to both sides of \eqref{Eq:GenSyl1} to obtain
\[
\Pi(sI_\nu - F) = (s I_n-A)\Pi - BH.
\]
Right multiplying this by $(sI_\nu-F)^{-1}$ and left-multiplying by $(sI_n - A)^{-1}$, we obtain
\[
(sI_n-A)^{-1}\Pi = \Pi (sI_\nu-F)^{-1} - (sI_n-A)^{-1}BH(sI_\nu-F)^{-1}.
\]
Let $\gamma$ be a Cauchy contour in the complex plane which encloses the eigenvalues of $F$ in its interior and excludes the eigenvalues of $A$. It follows by Cauchy's integral theorem and elementary application of the residue theorem that
\[
\int_{\gamma}(sI_n-A)^{-1}\,\mathrm{d} s = 0, \qquad \frac{1}{2\pi \j}\int_{\gamma}(sI_\nu-F)^{-1}\,\mathrm{d}s = I_\nu,
\]
and therefore
\[
\Pi = \frac{1}{2\pi {\j}}\int_{\gamma}(sI_n-A)^{-1}BH(sI_{\nu}-F)^{-1}\,\mathrm{d} s.
\]
It now follows by definition of $\Lp(H)$ that
\[
\Lp(H) = C\Pi + DH = \frac{1}{2\pi {\j}}\int_{\gamma}\hat{\Sigma}(s)H(sI_{\nu}-F)^{-1}\,\mathrm{d} s.
\]
Applying the residue theorem, the contour integral evaluates to
\[
\Lp(H) = \sum_{\lambda \in \{\text{poles of}\,\Gamma\,\text{inside}\,\gamma\}}\mathsf{Res}(\Gamma,\lambda),
\]
where $\Gamma(s) = \hat{\Sigma}(s)H(sI_{\nu}-F)^{-1} \in \complex^{p \times \nu}$ and $\mathsf{Res}(\Gamma,\lambda)$ denotes the residue of $\Gamma$ at the pole $\lambda$. By construction of $\gamma$, the poles of $\Gamma$ inside $\gamma$ are a subset of the eigenvalues of $F$, and pole-zero cancellations may occur. However, if some element of $\Gamma(s)$ has a removable singularity at $\lambda \in \mathrm{eig}(F)$, the contribution to the residue for that element will be zero, and we may therefore write
\[
\Lp(H) = \sum_{\lambda \in \eig(F)}\mathsf{Res}(\Gamma,\lambda).
\]
To evaluate the residues, we use the fact that (the singular portion of) the Laurent expansion of $(sI_\nu-F)^{-1}$ around the pole $\lambda_k \in \mathrm{eig}(F)$ is given by \cite[Equation (1.6)]{APC-DD:13}
\begin{equation}\label{Eq:LaurentF}
(sI_{\nu}-F)^{-1} = \sum_{j=0}^{m_k-1}\frac{1}{(s-\lambda_k)^{j+1}}X_{k,j}.
\end{equation}
{\tb At any point $s_0$ where $\hat{\Sigma}(s)$ is analytic, its series expansion may be expressed using moment matrices as \cite[Equation (1.4)]{APC-DD:13}
\begin{equation}\label{Eq:TaylorW}
\hat{\Sigma}(s) = \sum_{i=0}^{\infty}\nolimits (-1)^{i}\Mom_i(s_0)(s-s_0)^{i}.
\end{equation}
Since $\mathrm{eig}(A) \cap \mathrm{eig}(F) = \emptyset$, each element of $\hat{\Sigma}(s)$ is analytic at $s = \lambda_k$, so the expansion \eqref{Eq:TaylorW} is valid at $s_0 = \lambda_k$.} Combining \eqref{Eq:TaylorW} and \eqref{Eq:LaurentF}, the relevant portion of the Laurent series of $\Gamma$ around $\lambda_k$ is
\[
\begin{aligned}
\Gamma(s) &= \sum_{i=0}^{\infty}\sum_{j=0}^{m_k-1}\frac{1}{(s-\lambda_k)^{j+1-i}}(-1)^{i}\Mom_i(\lambda_k)H X_{k,j}
\end{aligned}
\]
from which we may read off that
\[
\mathsf{Res}(\Gamma,\lambda_k) = \sum_{j=0}^{m_k-1}(-1)^{j}\Mom_{j}(\lambda_k)H X_{k,j}
\]
which yields the result.
\end{proof}

\smallskip

The expressions \eqref{Eq:LfFreq} hint at a symmetry between the quantities $G\Lp(H)$ and $\Ld(G)H$. For completeness we make the following observations, the proofs of which are straightforward from \eqref{Eq:LfFreq}.

\begin{proposition}[\bf Symmetry between SSC Operators]\label{Prop:Symmetry}
The following statements hold:
\begin{enumerate}[(i)]
\item $G\Lp(H) = \Ld(G)H$ if the matrices $X_{k,j}$ and $G\Mom_j(\lambda_k)H$ commute for all $k \in \{1,\ldots,l\}$ and all $j \in \{0,\ldots,m_k-1\}$.
\item $G\Lp(H) = (\Ld(G)H)^{\sf T}$ if $F = F^{\sf T}$, $G$ has the form $G = \left[\begin{smallmatrix}g & g & \cdots & g\end{smallmatrix}\right]$ for any $g \in \real^\nu$, and $H = G^{\sf T}$.
\end{enumerate}
\end{proposition}

\section{Applications in Control, Estimation, and Model Reduction of Linear Systems}
\label{Sec:Applications}

This section details the application of the theory in Section \ref{Sec:SSC} to different linear control, estimation, and model reduction problems. While the presented designs are interesting in and of themselves in the context of linear systems, the broader goal is to systematically identify a large family of designs for future potential generalization to nonlinear systems, based on the core ideas in Theorem \ref{Thm:SSC}, Theorem \ref{Thm:SSC2}, and Corollary \ref{Cor:SSC3}. Two themes that will repeatedly appear are (i) \emph{recursive design}, wherein simpler sequential design problems are solved and combined to obtain an overall design, and (ii) \emph{low-gain design}, wherein a small tuning parameter is used to ensure {\tb closed-loop stability}. This suggests at least some of the presented design methods will admit generalizations to nonlinear systems, since recursive \cite{RS-MJ-PK:97} and low-gain designs \cite{JWSP:20a} are well-established in the nonlinear setting; this will be explored further in Section \ref{Sec:Nonlinear}. 

\medskip

\subsection{Stabilization of the Cascade $\Sigma \rightarrow \Sigma^{\prime}$}
\label{Sec:Cascade-Stab}

We begin by studying a recursive stabilizer design procedure for the cascaded system of Figure \ref{Fig:Cascade2}, described again by the equations
\begin{equation}\label{Eq:CascadeDual1}
\begin{aligned}
\begin{bmatrix}\dot{x} \\ \dot{\eta}\end{bmatrix} &= \begin{bmatrix}A & 0\\
GC & F
\end{bmatrix}\begin{bmatrix}x \\ \eta\end{bmatrix} + \begin{bmatrix}B \\ GD\end{bmatrix}u.
\end{aligned}
\end{equation}
The objective is to design a state feedback $u = K_{x}x + K_{\eta}\eta$ achieving exponential stabilization of the origin. Assume that $(A,B)$ is stabilizable, and let $K$ be such that (i) $\mathcal{A} \define A+BK$ is Hurwitz, and (ii)  $\eig(\mathcal{A}) \cap \eig(F) = \emptyset$. Consider the preliminary feedback $u = Kx + \tilde{u}$ applied to \eqref{Eq:CascadeDual1} leading to
\begin{equation}\label{Eq:ForwardingCascade1}
\begin{aligned}
\begin{bmatrix}\dot{x} \\ \dot{\eta}\end{bmatrix} &= \begin{bmatrix}\mathcal{A} & 0\\
G\mathcal{C} & F
\end{bmatrix}\begin{bmatrix}x \\ \eta\end{bmatrix} + \begin{bmatrix}B \\ GD\end{bmatrix}\tilde{u},
\end{aligned}
\end{equation}
where $\mathcal{C} \define C+DK$. This is precisely of the form \eqref{Eq:CascadeDual} under the substitutions $A \rightarrow \mathcal{A}$ and $C \rightarrow \mathcal{C}$ (cf. Remark \ref{Rem:StateFeedbackOutputInjection}). 

{\tb We can now describe six different paths for completing the design, based on our theory in Section \ref{Sec:SSC}; Table \ref{Tab:StabSummary} summarizes and compares the designs. For notational convenience in this table, we label strong and weak rank assumptions as:\\
(RS)\label{Item:NR1}  $\Ros_{\Sigma}(\lambda)$  full row rank for all $\lambda \in \eig(F)$;\\
(RW) \label{Item:NR2} $\Ros_{\Sigma}(\lambda)$  full row rank for all $\lambda \in \eig(F) \cap \complex_{\geq 0}$.

\begin{table}[h!]
\centering \color{black}
\renewcommand{\arraystretch}{1.4}
\begin{tabular}{|l|c|c|c|c|c|c|}
\hline
Design Path  & \ref{3A1a} & \ref{3A1b} & \ref{3A2a} & \ref{3A2b} & \ref{3A3a} & \ref{3A3b}\\
\hline
\hline
Rank condition & RW & RS & RW & RS & RS & RW\\
\hline
 $\mathrm{eig}(F) \subseteq \mathbb{C}_{\leq 0}$ & $\times$ & $\times$ & \checkmark & \checkmark & \checkmark & \checkmark \\
\hline
Feedback var. & $(x,\eta)$ & $(x,\eta)$ & $\eta$ & $\eta$ & $\eta$ & $\eta$\\
\hline
Assumes $p = m$ & $\times$ & \checkmark & $\times$ & \checkmark & $\times$ & \checkmark \\
\hline
Design uses & $M, \Ld$ & $M, \Ld$ & $\Ld$ & $\Ld$ & $\Lp$ & $\Lp$\\
\hline
Related Refs. & \cite{DA-LP-LM:22b} & -- & \cite{SM-DA-VA:22} & -- & \cite{LP:16} & --\\
\hline
\end{tabular}
\caption{\tb Comparison of six design pathways beginning from \eqref{Eq:ForwardingCascade1}. See Remark \ref{Rem:StabLit} for further literature notes.}
\label{Tab:StabSummary}
\end{table}

}

\smallskip


\subsubsection{Stabilization with $\Ld$}
\label{Sec:Cascade-Stab-Ld}

We begin with \eqref{Eq:ForwardingCascade1}. Following the same argument leading to \eqref{Eq:CascadeDualTransformed}, if we set $M = \Syld^{-1}(G\mathcal{C})$ and define the deviation variable $\zeta = \eta - Mx$, the dynamics \eqref{Eq:ForwardingCascade1} \emph{decouple} into the parallel interconnection
\begin{equation}\label{Eq:ForwardingTransformed}
\begin{aligned}
\begin{bmatrix}\dot{x} \\ \dot{\zeta}\end{bmatrix} &= \begin{bmatrix}\mathcal{A} & 0\\
0 & F
\end{bmatrix}\begin{bmatrix}x \\ \zeta\end{bmatrix} + \begin{bmatrix}B \\ \Ld(G)\end{bmatrix}\tilde{u}.
\end{aligned}
\end{equation}
Our theory now points to two pathways for completing the design, one exploiting a controllability property of $\Ld$, and the other exploiting an invertibility property of $\Ld$.

\paragraph{``Fix $G$, Design $K_{\eta}$''}\label{3A1a} If the conditions imposed in Theorem \ref{Thm:SSC} \eqref{SSCNewItm:6b} hold, then $(F,\Ld(G))$ is stabilizable. Choosing $K_{\eta}$ such that $F + \Ld(G)K_{\eta}$ is Hurwitz and applying the feedback $\tilde{u} = K_{\eta}\zeta$, \eqref{Eq:ForwardingTransformed} becomes
\begin{equation}\label{Eq:CascadeTriangularFinal}
\begin{aligned}
\begin{bmatrix}\dot{x} \\ \dot{\zeta}\end{bmatrix} &= \begin{bmatrix}\mathcal{A} & BK_{\eta}\\
0 & F + \Ld(G)K_{\eta}
\end{bmatrix}\begin{bmatrix}x \\ \zeta\end{bmatrix},
\end{aligned}
\end{equation}
which is a cascade of two  exponentially stable LTI systems, and thus is exponentially stable. In the original variables, the final state feedback is $u = (K - K_{\eta}M)x + K_{\eta}\eta$. This is essentially the \emph{forwarding} methodology; see, e.g., \cite{DA-LP-LM:22b}.

\paragraph{``Fix $K_{\eta}$, Design $G$''}\label{3A1b} Suppose that $p = m$ and that {\tb Theorem~\ref{Thm:SSC}\eqref{SSCNewItm:1} holds}. 
For \eqref{Eq:ForwardingTransformed}, let $\tilde{u} = K_{\eta}\zeta$ again be the chosen feedback, where $K_{\eta} \in \real^{m \times \nu}$ is any matrix such that $(F,K_{\eta})$ is detectable. Correspondingly, let $L$ be such that $F + LK_{\eta}$ is Hurwitz. 
{\tb Since $p=m$, Theorem \ref{Thm:SSC}\eqref{SSC2NewItm:5} yields that }$\Ld$ is invertible, and thus the linear equation $\Ld(G) = L$ possesses a unique solution $G = \Ld^{-1}(L)$. With this particular choice of $G$, one again obtains the closed-loop system \eqref{Eq:CascadeTriangularFinal} with the control $u = (K - K_{\eta}M)x + K_{\eta}\eta$, and the same statements hold.


\subsubsection{Low-Gain Stabilization with $\Ld$}
\label{Sec:Cascade-Stab-Ld-LowGain}

Under the additional assumption that $\eig(F) \subseteq \complex_{\leq 0}$, we now show that one may simplify the feedback design to $u = Kx + K_{\eta}\eta$. This simplification is of most interest when $A$ is Hurwitz, as in this case one may take $K = 0$ and thereby obtain a simple design which uses only the state $\eta$ of the driven system. This scenario arises, for example, in the design of so-called \emph{tuning regulators} \cite{EJD:76, LP:16, VN:21, LC-JWSP:23d}, which are minimal-order output-regulating controllers for internally stable LTI plants.

We return to the transformed system \eqref{Eq:ForwardingTransformed} and consider the feedback design $\tilde{u} = K_{\eta}\eta = K_{\eta}(\zeta + Mx)$, leading to the closed-loop system
\begin{equation}\label{Eq:TuningRegulatorClosedLoop}
\begin{bmatrix}
\dot{x} \\ \dot{\zeta}
\end{bmatrix} = \begin{bmatrix}\mathcal{A} + BK_{\eta}M & BK_{\eta}\\
\Ld(G)K_{\eta}M & F + \Ld(G)K_{\eta}\end{bmatrix}\begin{bmatrix}x \\ \zeta\end{bmatrix}.
\end{equation}
As \eqref{Eq:TuningRegulatorClosedLoop} no longer possesses a triangular structure, additional small-gain-type assumptions may be used to ensure closed-loop stability. {\tb The key idea is to introduce a small tunable parameter $\epsilon > 0$ somewhere within the loop.}

\paragraph{``Fix $G$, Design $K_{\eta}(\epsilon)$}\label{3A2a} Suppose again that the stabilizability conditions imposed in Theorem \ref{Thm:SSC} \eqref{SSCNewItm:6b} hold, hence $(F,\Ld(G))$ is stabilizable. Due to this and the fact that $\mathrm{eig}(F) \subseteq \complex_{\leq 0}$, one may in fact select $K_{\eta} \in \mathcal{E}(\real;\real^{m\times \nu})$ {\tb (see \emph{Notation})} such that $F + \Ld(G)K_{\eta}(\epsilon)$ is \emph{low-gain Hurwitz stable}\footnote{A continuous matrix-valued function $\map{\mathcal{A}}{\real_{\geq 0}}{\real^{n \times n}}$ is \emph{low-gain Hurwitz stable} if there exist constants $c,\epsilon^* > 0$ such that $\max_{\lambda \in \mathrm{eig}
(\mathcal{A}(\epsilon))} \mathrm{Re}[\lambda] \leq -c\epsilon$ for all $\epsilon \in [0,\epsilon^*)$; {\tb see \cite{LC-JWSP:23d} and Appendix~\ref{App:LGHS}.}}. With such a selection, a composite Lyapunov construction combining Lyapunov functions for $\mathcal{A}$ and $F + \Ld(G)K_{\eta}(\epsilon)$ can be used to show that \eqref{Eq:TuningRegulatorClosedLoop} is also low-gain Hurwitz stable; {\tb for completeness, this result may be found in Appendix \ref{App:LGHS}.} In sum, one is then guaranteed that for sufficiently small $\epsilon > 0$, the closed-loop system \eqref{Eq:TuningRegulatorClosedLoop} is exponentially stable with its dominant eigenvalue having a stability margin of $O(\epsilon)$. 

\paragraph{``Fix $K_{\eta}$, Design $G(\epsilon)$}\label{3A2b} Suppose that $p = m$ and that the conditions of Corollary \ref{Cor:SSC3} hold. Pick any matrix $K_{\eta} \in \real^{m \times \nu}$ such that $(F,K_{\eta})$ is detectable, and let $L \in \mathcal{E}(\real;\real^{\nu \times m})$ be such that $F + L(\epsilon)K_{\eta}$ is low-gain Hurwitz stable. By Corollary \ref{Cor:SSC3} \eqref{SSC3Itm:3}, $\Ld$ is invertible, and thus the \emph{linear} equation $\Ld(G) = L$ possesses a unique solution $G = \Ld^{-1}(L)$, which by linearity of $\Ld$ must also belong to $\mathcal{E}(\real;\real^{\nu \times p})$. We again arrive at the closed-loop system \eqref{Eq:TuningRegulatorClosedLoop}, and again {\tb a composite Lyapunov construction} may be applied to confirm low-gain Hurwitz stability.

\smallskip
\subsubsection{Low-Gain Stabilization with $\Lp$}
\label{Sec:LowGainStabLp}

Return to the cascaded system in the form \eqref{Eq:ForwardingCascade1}, and as in Section \ref{Sec:Cascade-Stab-Ld-LowGain}, suppose we are interested in the design of a simple feedback $\tilde{u} = K_{\eta}\eta$ depending \emph{only} on the state of the driven system. With $\Pi = \Sylp^{-1}(BK_{\eta})$ and the change of coordinate $\xi = x - \Pi \eta$, straightforward calculations show that the dynamics \eqref{Eq:ForwardingCascade1} take the form 
\begin{equation}\label{Eq:TuningRegulatorClosedLoop2}
\begin{bmatrix}
\dot{\xi} \\ \dot{\eta}
\end{bmatrix} = \begin{bmatrix}\mathcal{A} - \Pi G \mathcal{C} & -\Pi G \Lp(K_{\eta})\\
G\mathcal{C} & F + G\Lp(K_{\eta})\end{bmatrix}\begin{bmatrix}\xi \\ \eta\end{bmatrix}.
\end{equation}
The dynamics \eqref{Eq:TuningRegulatorClosedLoop2} are not dissimilar from \eqref{Eq:TuningRegulatorClosedLoop}, and we can develop two design procedures analogous to those in Section \ref{Sec:Cascade-Stab-Ld-LowGain} assuming again that $\eig(F) \subseteq \complex_{\leq 0}$.

\paragraph{``Fix $G$, Design $K_{\eta}(\epsilon)$}\label{3A3a} Suppose that Theorem \ref{Thm:SSC} \eqref{SSCNewItm:1} holds and that $G$ is such that $(F,G)$ is stabilizable. Now select any $Z \in \mathcal{E}(\real;\real^{p\times \nu})$ such that $F+GZ(\epsilon)$ is low-gain Hurwitz stable. Then by Theorem \ref{Thm:SSC}, $\Lp$ is a surjective linear operator, and thus there exists $K_{\eta} \in \mathcal{E}(\real;\real^{m \times \nu})$ such that $\Lp(K_{\eta}) = Z$. With this selection of $K_{\eta}$ in \eqref{Eq:TuningRegulatorClosedLoop2}, one can again employ composite Lyapunov arguments to establish that \eqref{Eq:TuningRegulatorClosedLoop2} is low-gain Hurwitz stable, completing the design. This particular design  was first presented in \cite{LC-JWSP:23d}.

\paragraph{``Fix $K_{\eta}$, Design $G(\epsilon)$''}\label{3A3b} Suppose that $p = m$, and pick any matrix $K_{\eta}$ such that $(F,K_{\eta})$ is detectable. If the conditions of Theorem \ref{Thm:SSC2}\eqref{SSC2NewItm:6b} hold, then $(F,\Lp(K_{\eta}))$ is detectable, and one may therefore design a gain $G \in \mathcal{E}(\real;\real^{\nu\times p})$ such that $F + G(\epsilon)\Lp(K_{\eta})$ is low-gain Hurwitz stable, with analogous arguments to before completing the stability proof. 

\begin{remark}[\bf Required Model Information]\label{Rem:ModelInfo}
It is important to compare how the presented designs use the model information $(A,B,C,D)$ of the system $\Sigma$. For simplicity, assume that $A$ is Hurwitz, in which case one may take $K = 0$ in all designs above. The design in Section \ref{Sec:Cascade-Stab-Ld} uses the solution $M$ of the Sylvester equation $\Syld(M) = GC$; this requires full knowledge of the plant $A$ matrix. In contrast, under the additional assumption that $\eig(F) \subseteq \complex_{\leq 0}$, the design in Section \ref{Sec:Cascade-Stab-Ld-LowGain} requires only knowledge of $\Ld$, and the design of Section \ref{Sec:LowGainStabLp} requires only knowledge of $\Lp$. It follows from Theorem \ref{Thm:Moments} that these latter two design approaches require only information about the \emph{moments} of the plant, which is a significant relaxation. The price paid for this minimal use of model information is degradation in performance, due to the low-gain character of the latter two designs. Identical comments will apply to our subsequent discussion of estimator design. \hfill \oprocend
\end{remark}

\begin{remark}[\bf Stabilization Literature Notes]\label{Rem:StabLit}
    Some of the procedures summarized in Section~\ref{Sec:Cascade-Stab} have appeared in literature on finite and infinite-dimensional linear systems.
    First, the design of Section \ref{Sec:Cascade-Stab-Ld} a) has appeared in the context of 
cascade stabilization via forwarding (or Sylvester approaches) and output regulation, and in \cite[Theorem 3.7]{VN:21} for ODE-PDE cascades {\tb (with particular attention to the problem of stabilization of plants in the presence of actuator and/or sensor dynamics)}. Supposing that the spectrum of $F$ is simple and lies on the imaginary axis, a standard choice of the gain $K_{\eta}$ is simply $K_\eta = - \Ld(G)^{\sf T}$, as shown in \cite{DA-LP-LM:22b}  for ODE-ODE cascades or \cite{SM-LB-DA:21} for ODE-PDE cascades. Under the same assumptions, the corresponding 
    ``small-gain'' procedure of Section \ref{Sec:Cascade-Stab-Ld-LowGain} a) has also been used in \cite{SM-DA-VA:22} for ODE-PDE cascades. Again for the case where the spectrum of $F$ is simple and lies on the imaginary axis, a version of the design in Section \ref{Sec:LowGainStabLp} a) was developed in \cite[Section 4]{LP:16} for output regulation of PDEs, wherein $K$ was chosen such that $\Lp(K)$ is full column rank, $K_{\eta}$ was selected as $K_{\eta} = \epsilon K$ for $\epsilon > 0$, and $G$ was set as $G = -\Lp(K)^*$, which is dual to the gain design in \cite{DA-LP-LM:22b,SM-DA-VA:22}.
    {\tb Finally, related ideas were also developed in 
 \cite[Chapter 3]{PVK-HKK:86} 
 in the context of singularly-perturbed systems, where feedback laws were obtained by approximating the solution of a non-symmetric algebraic Riccati equation (ARE) via a Sylvester equation.}
    \hfill \oprocend
\end{remark}

\subsection{Observation of the Cascade $\Sigma^{\prime} \rightarrow \Sigma$} 
\label{Sec:Cascade-Obsv}

We now examine an estimation problem for the cascaded system of Figure \ref{Fig:Cascade1}, described again by the equations
\[
\begin{aligned}
\begin{bmatrix}\dot{x} \\ \dot{\eta}\end{bmatrix} &= \begin{bmatrix}A & BH\\
0 & F
\end{bmatrix}\begin{bmatrix}x \\ \eta\end{bmatrix} + \begin{bmatrix}BJ \\ G\end{bmatrix}v, \ 
y = \begin{bmatrix}C & DH\end{bmatrix}\begin{bmatrix}x \\ \eta\end{bmatrix} + \begin{bmatrix}DJ\end{bmatrix}v.
\end{aligned}
\]
The objective is to design a state estimator based on the measurement $y$ of the driven system; we assume that the input signal $v$ is either known or absent. Consider the obvious candidate estimator
\[
\begin{aligned}
\begin{bmatrix}\dot{\hat{x}} \\ \dot{\hat{\eta}}\end{bmatrix} &= \begin{bmatrix}A & BH\\
0 & F
\end{bmatrix}\begin{bmatrix}\hat{x} \\ \hat{\eta}\end{bmatrix} + \begin{bmatrix}BJ \\ G\end{bmatrix}v - \begin{bmatrix}w_x \\ w_{\eta}\end{bmatrix}\\
\hat{y} &= \begin{bmatrix}C & DH\end{bmatrix}\begin{bmatrix}\hat{x} \\ \hat{\eta}\end{bmatrix} + \begin{bmatrix}DJ\end{bmatrix}v
\end{aligned}
\]
where $w_x,w_{\eta}$ are correcting inputs, to be designed next as linear functions of $y - \hat{y}$. Introducing the estimation errors $\tilde{x} = x - \hat{x}$, $\tilde{\eta} = \eta - \hat{\eta}$, and $\tilde{y} = y - \hat{y}$, the error dynamics are found to be
\begin{equation}\label{Eq:EstErrorDynamics}
\begin{aligned}
\begin{bmatrix}\dot{\tilde{x}} \\ \dot{\tilde{\eta}}\end{bmatrix} &= \begin{bmatrix}A & BH\\
0 & F
\end{bmatrix}\begin{bmatrix}\tilde{x} \\ \tilde{\eta}\end{bmatrix} + \begin{bmatrix}w_{x} \\ w_{\eta}\end{bmatrix}, \quad \tilde{y} = \begin{bmatrix}C & DH\end{bmatrix}\begin{bmatrix}\tilde{x} \\ \tilde{\eta}\end{bmatrix}.
\end{aligned}
\end{equation}

\smallskip


\subsubsection{Estimator Design with $\Lp$}
\label{Sec:EstDesignLp}

Assume $(A,C)$ is detectable, and let $L_x$ be such that (i) $\mathcal{A} := A+L_xC$ is Hurwitz and (ii) $\eig(\mathcal{A}) \cap \eig(F) = \emptyset$. Select $w_{x} = L_x\tilde{y} + \tilde{w}_x$ in \eqref{Eq:EstErrorDynamics}, leading to
\begin{equation}\label{Eq:ForwardingCascadeEst1}
\begin{aligned}
\begin{bmatrix}\dot{\tilde{x}} \\ \dot{\tilde{\eta}}\end{bmatrix} &= \begin{bmatrix}\mathcal{A} & \mathcal{B}H\\
0 & F
\end{bmatrix}\begin{bmatrix}\tilde{x} \\ \tilde{\eta}\end{bmatrix} + \begin{bmatrix}\tilde{w}_x \\ w_{\eta}\end{bmatrix}, \quad 
\tilde{y} = \begin{bmatrix}C & DH\end{bmatrix}\begin{bmatrix}\tilde{x} \\ \tilde{\eta}\end{bmatrix},
\end{aligned}
\end{equation}
where $\mathcal{B} = B+L_xD$ and where $\tilde{w}_x$ is still to be designed. Observe that \eqref{Eq:ForwardingCascadeEst1} has the same structure as \eqref{Eq:CascadePrimal} under the substitutions $A \rightarrow \mathcal{A}$ and $B \rightarrow \mathcal{B}$. Setting $\Pi = \Sylp^{-1}(\mathcal{B}H)$ and defining the change of state $\tilde{\xi} = \tilde{x} - \Pi\tilde{\eta}$, the error dynamics \eqref{Eq:ForwardingCascadeEst1} become
\begin{equation}\label{Eq:ObserverErrorForwarding}
\begin{aligned}
\begin{bmatrix}\dot{\tilde{\xi}} \\ \dot{\tilde{\eta}}\end{bmatrix} &= \begin{bmatrix}\mathcal{A} & 0\\
0 & F
\end{bmatrix}\begin{bmatrix}\tilde{\xi} \\ \tilde{\eta}\end{bmatrix} + \begin{bmatrix}\tilde{w}_x - \Pi w_{\eta}\\ w_{\eta}\end{bmatrix}, \quad 
\tilde{y} = \begin{bmatrix}C & \Lp(H)\end{bmatrix}\begin{bmatrix}\tilde{\xi} \\ \tilde{\eta}\end{bmatrix}.
\end{aligned}
\end{equation}
Precisely mirroring the development in Section \ref{Sec:Cascade-Stab-Ld}, the theory of Section \ref{Sec:SSC} again points to two pathways for completing the design.

\paragraph{``Fix $H$, Design $L_{\eta}$''} If the conditions imposed in Theorem \ref{Thm:SSC2}\eqref{SSC2NewItm:6b} hold, then $(F,\Lp(H))$ is detectable, and therefore there exists $L_{\eta} \in \real^{\nu \times p}$ such that $F + L_{\eta}\Lp(H)$ is Hurwitz. With the choices $w_{\eta} = L_{\eta}\tilde{y}$ and $\tilde{w}_x = \Pi w_{\eta}$, the error dynamics reduce to
\begin{equation}\label{Eq:ObserverCascadeTriangular}
\begin{aligned}
\begin{bmatrix}\dot{\tilde{\xi}} \\ \dot{\tilde{\eta}}\end{bmatrix} &= \begin{bmatrix}\mathcal{A} & 0\\
L_{\eta}C & F+L_{\eta}\Lp(H)
\end{bmatrix}\begin{bmatrix}\tilde{\xi} \\ \tilde{\eta}\end{bmatrix},
\end{aligned}
\end{equation}
which is cascade of linear exponentially stable systems, and is thus exponentially stable. We therefore obtain the final observer design (in the original coordinates) as
\begin{equation}\label{Eq:FinalObserver}
\begin{aligned}
\begin{bmatrix}\dot{\hat{x}} \\ \dot{\hat{\eta}}\end{bmatrix} &= \begin{bmatrix}A & BH\\
0 & F
\end{bmatrix}\begin{bmatrix}\hat{x} \\ \hat{\eta}\end{bmatrix} + \begin{bmatrix}BJ \\ G\end{bmatrix}v + \begin{bmatrix}L_{x} + \Pi L_{\eta} \\ L_{\eta}\end{bmatrix}(\hat{y} - y)\\
\hat{y} &= \begin{bmatrix}C & DH\end{bmatrix}\begin{bmatrix}\hat{x} \\ \hat{\eta}\end{bmatrix} + \begin{bmatrix}DJ\end{bmatrix}v.
\end{aligned}
\end{equation}

\paragraph{``Fix $L_{\eta}$, Design $H$''} Assume now that $p = m$ and that the conditions of Corollary \ref{Cor:SSC3} hold. Pick any matrix $L_{\eta}$ such that $(F,L_{\eta})$ is stabilizable, and let $K$ be such that $F + L_{\eta}K$ is Hurwitz. By Corollary \ref{Cor:SSC3} \eqref{SSC3Itm:3}, $\Lp$ is invertible, and thus the linear equation $\Lp(H) = K$ possesses a unique solution $H = \Lp^{-1}(K)$. With this particular choice of $H$, one again obtains the triangular error dynamics \eqref{Eq:ObserverCascadeTriangular} and the same stability conclusions hold.

\medskip


\subsubsection{Low-Gain Estimator Design with $\Lp$ (Tuning Estimators)}
\label{Sec:LowGainEstimatorLp}

Mirroring the ideas in Section \ref{Sec:Cascade-Stab-Ld-LowGain}, we now relax some aspects of the previous estimator design procedure using low-gain methods, again under the further assumption that $\eig(F) \subseteq \complex_{\leq 0}$. We will call these designs \emph{tuning estimators}, as they are dual to idea of a tuning regulator as mentioned in Section \ref{Sec:Cascade-Stab-Ld-LowGain}.

To begin, in \eqref{Eq:ObserverErrorForwarding} we select $\tilde{w}_x = 0$, which will eliminate the need for computation of the matrix $\Pi$ in the design. If we further select $w_{\eta} = L_{\eta}\tilde{y}$, we obtain the non-triangular error dynamics
\begin{equation}\label{Eq:NonTriangularEstimator1}
\begin{aligned}
\begin{bmatrix}\dot{\tilde{\xi}} \\ \dot{\tilde{\eta}}\end{bmatrix} &= \begin{bmatrix}\mathcal{A}-\Pi L_{\eta}C & -\Pi L_{\eta} \Lp(H)\\
L_{\eta}C & F+L_{\eta}\Lp(H)
\end{bmatrix}\begin{bmatrix}\tilde{\xi} \\ \tilde{\eta}\end{bmatrix},
\end{aligned}
\end{equation}
which is non-triangular but directly analogous to \eqref{Eq:TuningRegulatorClosedLoop}. Our theory now enables two ways to complete the design.

\paragraph{``Fix $H$, Design $L_{\eta}(\epsilon)$''} If the conditions of Theorem \ref{Thm:SSC2}\eqref{SSC2NewItm:6b} hold (i.e., $(F,H)$ detectable and $\Ros_{\Sigma}(\lambda)$ full column rank for all $\lambda \in \eig(F) \intersection \complex_{\geq 0}$), then $(F,\Lp(H))$ is detectable, and therefore there exists $L_{\eta} \in \mathcal{E}(\real;\real^{\nu\times p})$ such that $F + L_{\eta}(\epsilon)\Lp(H)$ is low-gain Hurwitz stable. As described previously, a composite Lyapunov construction can now be used to establish that the overall error dynamics \eqref{Eq:NonTriangularEstimator1} are low-gain Hurwitz stable. 

\paragraph{``Fix $L_{\eta}$, Design $H(\epsilon)$''} Assume now that $p = m$ and that the conditions of Corollary \ref{Cor:SSC3} hold. Pick any matrix $L_{\eta}$ such that $(F,L_{\eta})$ is stabilizable, and let $Y \in \mathcal{E}(\real;\real^{p\times \nu})$ be such that $F + L_{\eta}Y(\epsilon)$ is low-gain Hurwitz stable. By Corollary \ref{Cor:SSC3} \eqref{SSC3Itm:3}, $\Lp$ is invertible, and thus the linear equation $\Lp(H) = Y$ possesses a unique solution $H = \Lp^{-1}(Y) \in \mathcal{E}(\real;\real^{m \times \nu})$. Making these choices of $H$ and $L_{\eta}$ in \eqref{Eq:NonTriangularEstimator1} complete the design.

\smallskip
\subsubsection{Low-Gain Estimator Design with $\Ld$}
\label{Sec:LowGainEstimatorLd}

The dual SSC operator $\Ld$ can also be used for low-gain estimator design. We return to the estimation error dynamics in the original coordinates \eqref{Eq:ForwardingCascadeEst1}, and set $\tilde{w}_x = 0$. We are again interested in design $w_{\eta} = L_{\eta}\tilde{y}$ for some $L_{\eta}$ to be determined. With $M = \Syld^{-1}(L_{\eta}C)$ and the change of coordinate $\tilde{\zeta} = \tilde{\eta} - M\tilde{x}$, the error dynamics \eqref{Eq:ForwardingCascadeEst1} become
\begin{equation}\label{Eq:NonTriangularEstimator2}
\begin{aligned}
\begin{bmatrix}\dot{\tilde{x}} \\ \dot{\tilde{\zeta}}\end{bmatrix} &= \begin{bmatrix}\mathcal{A} + \mathcal{B}HM & \mathcal{B}H\\
\Ld(L_{\eta})HM & F+\Ld(L_{\eta})H
\end{bmatrix}\begin{bmatrix}\tilde{x} \\ \tilde{\zeta}\end{bmatrix},
\end{aligned}
\end{equation}
which is of course quite similar to \eqref{Eq:NonTriangularEstimator1}. If $\eig(F) \subseteq \complex_{\geq 0}$, we once again have two paths for completing the design. In brief:
\begin{enumerate}[(a)]
\item ``\emph{Fix $H$, Design $L_{\eta}(\epsilon)$}'': Fix $H$ such that $(F,H)$ detectable, pick $Z \in \mathcal{E}(\real;\real^{\nu \times m})$ such that $F+Z(\epsilon)H$ is low-gain Hurwitz stable, and exploit surjectivity of $\Ld$ from Theorem \ref{Thm:SSC2} to obtain a gain $L_{\eta} \in \mathcal{E}$ satisfying $Z = \Ld(L_{\eta})$. 
\item ``\emph{Fix $L_{\eta}$, Design $H(\epsilon)$}'': If $p = m$, fix $L_{\eta}$ such that $(F,L_{\eta})$ is stabilizable, and exploit stabilizability of $(F,\Ld(L_{\eta}))$ by Corollary \ref{Cor:SSC3} to design $H \in \mathcal{E}$ such that $F + \Ld(L_{\eta})H(\epsilon)$ is low-gain Hurwitz stable.
\end{enumerate}



\begin{remark}[\bf Implementation of Tuning Estimators]\label{Rem:TuningEstimator}
In the low-gain design procedures of Sections \ref{Sec:LowGainEstimatorLp} and {\tb \ref{Sec:LowGainEstimatorLd}}, when $A$ is known to be Hurwitz, we may select $L_x = 0$ and the computed observer \eqref{Eq:FinalObserver} can be expressed as
\begin{equation}\label{Eq:TuningObserver}
\begin{aligned}
\dot{\hat{x}} &= A\hat{x} + B\hat{z},\quad &\hat{y} &= C\hat{x} + D\hat{z},\\
\dot{\hat{\eta}} &= F\hat{\eta} + Gv + L_{\eta}(\hat{y}-y),\quad &\hat{z} &= H\hat{\eta}+Jv.
\end{aligned}
\end{equation}
The first line of \eqref{Eq:TuningObserver} can be interpreted as an open-loop input-output simulation of the plant $\Sigma = (A,B,C,D)$, wherein the estimated signal $\hat{z}$ is used to predict the plant output, yielding $\hat{y}$. The prediction $\hat{y}$ is then used in the second line of \eqref{Eq:TuningObserver} to produce the estimate $\hat{\eta}$. The tunable gain in the design is $L_{\eta}$, which is obtained via the operator $\Lp$, and hence, may be designed using only the moments of the plant $\Sigma$ as model information. This leads to the following idea: the first line of \eqref{Eq:TuningObserver} may be implemented using \emph{any} methodology that enables forward input-output simulation of the plant $\Sigma$ (e.g., via a non-parametric data-driven model or a digital twin \cite{LL-CDP-PT-NM:24}), with the second line of \eqref{Eq:TuningObserver} then processing that output, along with the true measurement, to estimate $\eta$. \hfill \oprocend
\end{remark}

\begin{remark}[\bf Observation Literature Notes]\label{Rem:Ob}
{\tb In the context of observers for cascaded systems involving actuator and sensor dynamics}, the design of Section \ref{Sec:EstDesignLp} a) has been investigated for infinite-dimensional linear systems in \cite[Theorem 4.4]{VN:21}. {\tb To the best of our knowledge, all other design procedures are novel.}
  {\tb Finally, it is worth highlighting that a direct application of this theory is the case of disturbance estimators (or extended state observers)
  in which the disturbance is generated by a ``known'' model; see, e.g. \cite[Section II.C]{WHC-JY-LG-SL:15} or \cite[Section 5]{BRA-IBF:20}.} 
    {\tb As for the stabilization case, 
    Sylvester equations obtained by linearizing non-symmetric AREs were used in 
 \cite[Chapter 3]{PVK-HKK:86} 
to analyze the stability of a Kalman-Bucy steady-state filter for composite singularly-perturbed systems. }
       \hfill \oprocend
\end{remark}


\subsection{Structural Properties of Reduced-Order Models} 
\label{Sec:Cascade-MOR}



Finally, we illustrate the applications of our theory to the problem of model order reduction by \emph{moment matching}. Roughly speaking, the objective is to begin with a plant $\Sigma \define (A,B,C,D)$ and obtain a new plant $\rom{\Sigma} \define (\rom{A},\rom{B},\rom{C},\rom{D})$ of \textit{lower order that matches the moments of $\Sigma$ at $\Sigma^{\prime} \define (F,G,H)$}, i.e., that matches the moments\footnote{The notation is as defined in Section~\ref{Sec:Moments}: $\lambda_k$ for $k \in \{1,\dots,l\}$ denote distinct eigenvalues of $F$, $\Mom_j(\lambda_k)$ denotes the moment at $\lambda_k$ of order $j \in \{0,\ldots,m_k\}$ where $m_k$ is the algebraic multiplicity of $\lambda_k$, and $h_i$ (resp. $g_i$), with $i \in \{1,\dots,\nu\}$, denote the columns of $H$ (resp. rows of $G$).} $\Mom_j(\lambda_k)$ of the original plant at selected interpolation points $\lambda_k \in \eig(F)$ along selected right directions $h_i\in\mathbb{C}^m$ (the columns of $H$) and/or left directions $g_i\in\mathbb{C}^p$ (the rows of $G$). More precisely, if we denote the moments of $\rom{\Sigma}$ by $\hat{\Mom}_j(\lambda_k)$, we aim to construct $\rom{\Sigma}$ such that a number of conditions $\Mom_j(\lambda_k)h_i = \hat{\Mom}_j(\lambda_k)h_i$ and/or $g_i\Mom_j(\lambda_k) = g_i \hat{\Mom}_j(\lambda_k)$ hold, for $i,j,k$ over certain index sets.
This problem has been solved in the literature by different approaches{\tb, such as using the interpolation theory \cite{ACA:05} and the Loewner framework \cite{AJM-ACA:07}. The following presentation is based on the so-called ``time-domain moment matching'', which is more recently known as ``interconnection-based model reduction'' \cite{GS-AA:24}.} In the rest of this section we assume that $\Sigma$ is minimal, that $(F,H)$ is observable, and that $(F,G)$ is controllable; for reasons explained in detail in \cite{GS-AA:17}, these properties are always assumed in the model reduction literature. 

One solution to this problem, which considers only right direction matching, is given by 
\begin{equation}
\label{ROM_class_LP}
\rom{A} \define F -  \rom{B}H,\qquad
\rom{C} \define \Lp(H)-\rom{D}H,
\end{equation}
for any $\rom{B}$ such that $\eig(\rom{A}) \cap \eig(F) = \emptyset$ and any $\rom{D}$. The family~(\ref{ROM_class_LP}) parameterized in $(\rom{B},\rom{D})$ identifies all the reduced-order models of order $\nu$ which match (a linear combination of) the moments of $\Sigma$ computed at the eigenvalues of the matrix $F$ \textit{along} the directions identified by the columns of $H$, 
see \cite[Lemma 3]{MFS-GS-AYP-AP-NVDW:23} for details. When $\Sigma$ is single-input single-output (SISO), then this property reduces exactly to matching the (scalar) moments $\Mom_j(\lambda_k)$ computed at the eigenvalues of the matrix $F$, and in this case $H$ plays no role. The interpretation is that a reduced-order model by moment matching matches the frequency response (and possibly its derivatives) of $\Sigma$ at the frequencies encoded in $F$. For all the above reasons, in the literature the objects $\Lp(H)$ and $\Ld(G)$ themselves are called, with abuse, \textit{moments} of $\Sigma$ at $\eig(F)$. It is worth stressing however that while the elements of $\Lp(H)$ and $\Ld(G)$ are indeed exactly the moments in the SISO case, in the MIMO case these are linear combinations of the moments along the directions $H$ and $G$, respectively (see \cite[Equation (7)]{MFS-GS-AYP-AP-NVDW:23} for the exact relation). 
Consequently, in the MIMO case the frequency response interpretation depends on the selection of the directions $H$ and $G$, see \cite[Section 2.4]{MFS-GS-AYP-AP-NVDW:23}.

Another solution to the moment matching problem, which considers only left-direction matching \cite{GS-AA:24}, is given by
\begin{equation}
\label{ROM_class_Ld}
\rom{A}\define F - G \rom{C},\qquad 
\rom{B}\define \Ld(G)-G\rom{D},
\end{equation}
for any $\rom{C}$ such that $\eig(\rom{A}) \cap \eig(F) = \emptyset$ and any $\rom{D}$. Any model in this family parameterized in $(\rom{C},\rom{D})$  matches (a linear combination of) the moments of $\Sigma$ computed at the eigenvalues of the matrix $F$ \textit{along} the directions identified by the rows of $G$. In the following, we reinterpret the constructions of these reduced-order models using the interconnections in Figures~\ref{Fig:Cascade1} and \ref{Fig:Cascade2} and establish their structural system-theoretic properties. 


\smallskip
\subsubsection{Reduction with $\Lp$ using the Cascade $\Sigma^{\prime} \rightarrow \Sigma$ 
}
\label{Sec:Cascade2-Red-Lp}

Consider the cascaded system of Figure \ref{Fig:Cascade1}. As described in Section~\ref{Sec:SSC_DI}, the matrix $\Pi = \Sylp^{-1}(BH)$ defines the invariant subspace $\setdef{(x,\eta)}{x = \Pi\eta}$ for the dynamics \eqref{Eq:CascadePrimal}. Recall also that the unforced dynamics on the invariant subspace are simply described by
\begin{equation}
\dot{\eta} = F\eta, \qquad y = \Lp(H)\eta.
\end{equation}
Moreover, if $A$ is Hurwitz, then trajectories of \eqref{Eq:CascadePrimalTransformed} converge to this invariant subspace, and $\Lp(H)$ describes the steady-state gain relating the observation $y$ to the state $\eta$ of the driving system. Thus, matching the moments $\Lp(H)$ means matching the steady-state output response of $\Sigma$ for input signals generated by $\Sigma^{\prime}$. The problem of moment matching is then solved if $\rom{\Sigma}$ is designed such that the matching condition 
\begin{equation}\label{Eq:ROMMatching2}
\rom{\Lp}(H)=\Lp(H)
\end{equation}
holds, where $\rom{\Lp}(H) \define \rom{C}\rom{\Pi} + \rom{D}H = \rom{C}\rom{\Sylp}^{-1}(\rom{B}H) + \rom{D}H$, and where the associated primal Sylvester operator is defined as $\rom{\Sylp}(\rom{\Pi}) = \rom{\Pi}F - \rom{A}\rom{\Pi}$. Put differently, we want to determine $(\rom{A},\rom{B},\rom{C},\rom{D})$ such that $\rom{\Pi} := \rom{\Sylp}^{-1}(\rom{B}H)$ 
and the condition~(\ref{Eq:ROMMatching2}) holds. Straightforward computations show that this is achieved by the selection
\begin{equation}
\label{ROM_class_LP_gen}
\begin{aligned}
\rom{A}&=\rom{\Pi}F\rom{\Pi}^{-1} -  \rom{B}H\rom{\Pi}^{-1},\\
\rom{C}&=\left(\Lp(H)-\rom{D}H\right)\rom{\Pi}^{-1},
\end{aligned}
\end{equation}
for any invertible $\rom{\Pi}$, any $\rom{B}$ such that $\eig(\rom{A}) \cap \eig(F) = \emptyset$ and any $\rom{D}$.
Note that (\ref{ROM_class_LP_gen}) and (\ref{ROM_class_LP}) are similar representations, with change of coordinates given by $\rom{\Pi}$. If $\rom{B}$ is selected such that $\rom{A}$ is Hurwitz, then $\rom{\Sigma}$ and $\Sigma$ have the same steady-state output response for input signals generated by $\Sigma^{\prime}$. While the analysis above has been originally derived in \cite{AA:10}, the ideas in Theorem \ref{Thm:SSC2} may now be applied to study observability of the reduced-order model.

\begin{theorem}
[Observability of the Reduced-Order Model]\label{Cor:ROMObs}
Let $\rom{A}$ and $\rom{C}$ be defined as in (\ref{ROM_class_LP}) or (\ref{ROM_class_LP_gen}).
Suppose that $(F,H)$ is observable, that $\eig(\rom{A}) \cap \eig(F) = \emptyset$, and consider the operator $\Lp$ defined in \eqref{Eq:GenL}. For any $\rom{B}$ and $\rom{D}$ such that 
$$
\begin{bmatrix} F - \lambda I_\nu & \rom{B}\\ \Lp(H) & \rom{D}\end{bmatrix}
$$ 
has full column rank for all $\lambda \in \eig(\rom{A})$, the pair $(\rom{A},\rom{C})$ is observable.
\end{theorem}




\begin{proof}
As the systems (\ref{ROM_class_LP}) and (\ref{ROM_class_LP_gen}) are similar, we examine observability of (\ref{ROM_class_LP}). 
For (\ref{ROM_class_LP}), consider the associated Sylvester equation $\rom{\Pi}F - \rom{A}\rom{\Pi} = \rom{B}H$, or more explicitly
\begin{equation}
\label{eq-sylrom}
\rom{\Pi}F - (F-\rom{B}H)\rom{\Pi} = \rom{B}H.
\end{equation}
Since $\eig(\rom{A}) \cap \eig(F) = \emptyset$, the unique solution of \eqref{eq-sylrom} is $\rom{\Pi}=I_\nu$. Define the primal SCC operator for this system as $\rom{\Lp}(H)\define \rom{C}\rom{\Pi} + \rom{D} H$. Substituting $\rom{\Pi} = I_{\nu}$ and $\rom{C}$ from \eqref{ROM_class_LP} and rearranging, we obtain 
%
%
\begin{equation}
\label{eq-CpRomProof}
\rom{\Lp}(H)-\rom{D} H = \rom{C} = \Lp(H) -\rom{D}H.
\end{equation}
Select $\lambda \in \eig(F-\rom{B}H)$ with $f \in \complex^{\nu}$ a right-eigenvector of $F-\rom{B}H$ associated with $\lambda$. Right-multiplying \eqref{eq-sylrom} and \eqref{eq-CpRomProof} by $f$, and exploiting that $\rom{\Pi}=I_{\nu}$, we obtain
\begin{equation}\label{Eq:TempRewritten2rom}
\begin{bmatrix} F - \lambda I_\nu & \rom{B}\\ \Lp(H) & \rom{D}\end{bmatrix} \begin{bmatrix} f \\ -Hf \end{bmatrix} = \begin{bmatrix}0 \\ \rom{C}f\end{bmatrix}.
\end{equation}
Since $(F,H)$ is observable, then $(F-\rom{B}H,H)$ is observable because $F-\rom{B}H$ is an output injection of $H$. If $\lambda$ is observable for $(F-\rom{B}H,H)$
and $\Ros_{\rom{\Sigma}}(\lambda)$ has full column rank, then $Hf \neq 0$ and the left-hand side of \eqref{Eq:TempRewritten2rom} cannot be zero, so we conclude that $\rom{C}f \neq 0$; since $f \in \ker(\lambda I_{\nu}-(F-\rom{B}H))$ was arbitrary, observability of $\lambda$ for the pair $(F-\rom{B}H,\rom{C})$ follows from the eigenvector test.
\end{proof}



Note that in the model order reduction literature, the assumptions that $\eig(\rom{A}) \cap \eig(F) = \emptyset$ and $(F,H)$ is observable hold always. The first holds by construction of the reduced-order model while the second is assumed without loss of generality when constructing $(F,H)$ from the interpolation data (target points and directions). Thus, practically, only the rank condition in Theorem~\ref{Cor:ROMObs} needs to be verified. Moreover, note that this rank condition applies to the entire family of models as several $\rom{B}$ and $\rom{D}$ can satisfy the condition. 




\smallskip

\smallskip
\subsubsection{Reduction with $\Ld$ using the Cascade $\Sigma \rightarrow \Sigma^{\prime}$}
\label{Sec:Cascade1-Red-Ld}



Consider the cascaded system of Figure \ref{Fig:Cascade2}, described by the equations \eqref{Eq:CascadeDual}, and subsequently by \begin{equation}\label{Eq:ForwardingTransformedMOR}
\begin{aligned}
\begin{bmatrix}\dot{x} \\ \dot{\zeta}\end{bmatrix} &= \begin{bmatrix}A & 0\\
0 & F
\end{bmatrix}\begin{bmatrix}x \\ \zeta\end{bmatrix} + \begin{bmatrix}B \\ \Ld(G)\end{bmatrix}u.
\end{aligned}
\end{equation}
after the coordinate transformation $\zeta := \eta - Mx$ where $M = \Syld^{-1}(GC)$. As observed in Section~\ref{Sec:SSC_DI}, the matrix $\Ld(G)$ is the input matrix for the dynamics of the deviation variable $\zeta$, and thus influences how the control $u$ impacts the deviation from the invariant subspace. Moreover, if $A$ is Hurwitz, $x(0)=0$, and $\zeta(0)=0$, then the impulse response matrix of the interconnection~(\ref{Eq:ForwardingTransformedMOR}) is given by
$$
\zeta(t) = e^{Ft} \Ld(G) \mathsf{1}(t),
$$
where $\mathsf{1}(t)$ denotes the unit step signal.
Thus, matching the moments $\Ld(G)$ can be interpreted as matching the impulse response of $\Sigma$ \emph{filtered} through $\Sigma'$. The problem of moment matching is then solved if $\rom{\Sigma}$ is designed such that the matching condition
\begin{equation}\label{Eq:ROMMatching1}
\rom{\Ld}(G)=\Ld(G)
\end{equation}
holds where $\rom{\Ld}(G) \define -\rom{M}\rom{B} + G\rom{D} = -\rom{\Syld}^{-1}(G\rom{C})\rom{C} + G\rom{D}$. In summary, we want to determine $(\rom{A},\rom{B},\rom{C},\rom{D})$ such that $\rom{M} := \rom{\Syld}^{-1}(G\rom{C})$ and the condition~(\ref{Eq:ROMMatching1}) holds. Straightforward computations show that this is achieved by the selection
\begin{equation}
\label{ROM_class_Ld_gen}
\begin{aligned}
\rom{A}&=\rom{M}^{-1}F\rom{M} + \rom{M}^{-1}G \rom{C},\\
\rom{B}&=-\rom{M}^{-1}\left(\Ld(G)-G\rom{D}\right),
\end{aligned}
\end{equation}
for any invertible $\rom{M}$, any $\rom{C}$ such that $\eig(\rom{A}) \cap \eig(F) = \emptyset$ and any $\rom{D}$. Note that (\ref{ROM_class_Ld_gen}) and (\ref{ROM_class_Ld}) are similar representations, with change of coordinates given by $\rom{M}$. If $\rom{C}$ is selected such that $\rom{A}$ is Hurwitz, then $\rom{\Sigma}$ and $\Sigma$ have the same impulse response filtered through $\Sigma^{\prime}$. We can now leverage the ideas in Theorem \ref{Thm:SSC} to obtain the following controllability result, which is dual to the observability result of Theorem~\ref{Cor:ROMObs}.

\begin{theorem}
[Controllability of the Reduced-Order Model]\label{Cor:ROMCtrb}
Let $\rom{A}$ and $\rom{B}$ be defined in (\ref{ROM_class_Ld}) or (\ref{ROM_class_Ld_gen}). 
Suppose that $(F,G)$ is controllable, $\eig(\rom{A}) \cap \eig(F) = \emptyset$, and consider the operator $\Ld$ defined in \eqref{Eq:GenL}. For any $\rom{C}$ and $\rom{D}$ such that 
$$
\begin{bmatrix}F-\lambda I_{\nu}  & \Ld(G)\\
\rom{C} & \rom{D}\end{bmatrix}
$$ 
has full row rank for all $\lambda \in \eig(\rom{A})$, the pair $(\rom{A},\rom{B})$ is controllable.
\end{theorem}




\begin{remark}[\bf Two-Sided Moment Matching]\label{Rem:MOR2s} A natural question that arises is whether it is possible to select the free parameters in \eqref{ROM_class_LP} or \eqref{ROM_class_Ld} so that additional moments are matched. This question has an affirmative answers as long as the additional moments are computed at new interpolation points. Specifically, consider then the change of notation $(F,H) \to (F^H,H)$ for the family \eqref{ROM_class_LP} and $(F,G) \to (F^G,G)$ for the family \eqref{ROM_class_Ld}, with $\eig(F^H) \cap \eig(F^G) = \emptyset$. Then the selection
$$
\rom{B} \define (M\Pi)^{-1} M B \quad (\text{resp.} \rom{C} \define C\Pi (M\Pi)^{-1})
$$
is such that the family \eqref{ROM_class_LP} (resp. \eqref{ROM_class_Ld}) matches the right interpolation data $(F^H,H)$ and the left interpolation data $(F^G,G)$, simultaneously. This result, which is due to \cite{TCI:15}, corresponds to a two-sided interconnection $\Sigma^{\prime} \rightarrow \Sigma \rightarrow \Sigma^{\prime}$; the development of a complete interpretation of this two-sided interconnection via SSC operators is beyond the scope of this section.
       \hfill \oprocend
\end{remark}

\begin{remark}[\bf Model Reduction Literature Notes]\label{Rem:MOR}
The relation between moments and the solution of a Sylvester equation was pointed out in \cite{KAG-AV-PVD:04,KAG-AV-PVD:06}. The relation between moments and the steady-state response of $\Sigma^{\prime} \to \Sigma$ was recognized in \cite{Ast:07b}. 
The interconnection $\Sigma \to \Sigma^{\prime}$ was firstly introduced in \cite{AA:10a}, while the two-sided interconnection $\Sigma^{\prime} \to \Sigma \to \Sigma^{\prime}$ was given in \cite{TCI:15}. A completely equivalent framework to solve the (two-sided) moment matching problem (for zero-order moments) is represented by the Loewner framework \cite{AJM-ACA:07}. An interconnection interpretation of the Loewner framework was given in \cite{JDS-AA:19}. 
We highlight again, that the main novelty of this section is the characterization of the structural properties of families of reduced-order models based on non-resonance-type conditions, as an application of the results of Section~\ref{Sec:SSC}.
       \hfill \oprocend
\end{remark}


\subsection{Example: Four-Tank Flow Control System}
 
 {\tb We illustrate the stabilization and estimation design procedures by applying them to a (linearized) four-tank flow control system; see \cite{KHJ:00} for the description of the system and an illustration.}

{\tb First, the stabilization results of Section \ref{Sec:Cascade-Stab} will be applied for the problem of output regulation, i.e., combined reference tracking and disturbance rejection control; see \cite[Chapter 4]{AI:17} for an overview of linear output regulation.} In this context, the system $\Sigma$ in Figure \ref{Fig:Cascade2} is the four-tank system, which has two control inputs $(u_1,u_2)$ (valve flow rates) and two outputs $(e_1,e_2)$ (water level tracking errors in the lower tanks) which should be regulated to zero. The system $\Sigma^{\prime}$ in Figure \ref{Fig:Cascade2} models a \emph{post-processing internal model} {\tb \cite[Eq. (4.28)]{AI:17}}, designed here with
\begin{subequations}
\begin{align}
F &= \mathrm{blkdiag}\left(\left[\begin{smallmatrix}
0
\end{smallmatrix}\right],\left[\begin{smallmatrix}
0 & 1\\
-\omega_1^2 & 0
\end{smallmatrix}\right],\left[\begin{smallmatrix}
0 & 1\\
-\omega_2^2 & 0
\end{smallmatrix}\right]\right) \otimes I_2\\
G &= \mathrm{col}(1,0,1,0,1)\otimes I_2
\end{align}
\end{subequations}
to reject constant disturbances and disturbances at frequencies $\omega_1 = 0.001$ rad/s and $\omega_2 = 0.005$ rad/s. {\tb An external disturbance is present, namely an exogenous water flow into the second upper tank given by $d(t) = 20 + 20\sin(\omega_1t) + 30\sin(\omega_2 t)$ for $t \geq 0$.} {\tb Figure \ref{Fig:Sim1} shows the closed-loop response to such a disturbance when stabilizers are designed using the six methods of Section \ref{Sec:Cascade-Stab}, along with a baseline method of LQR applied to the composite system \eqref{Eq:CascadeDual1}; tuning choices have been made here so that the plots are distinguishable}.

 {\tb To illustrate the estimator designs of Section \ref{Sec:Cascade-Obsv}, the example is slightly modified. The control inputs are kept at zero, and the goal is now to estimate the unmeasured inflow disturbance entering the upper tank. In the context of Figure \ref{Fig:Cascade1}, $\Sigma$ is again the four-tank system with measurement $y = e_2 \in \real$ (the height error in the second lower tank) and input $u = d \in \real$, the flow disturbance entering the upper tank. The system $\Sigma^{\prime}$ is the autonomous system $\dot{\eta} = F\eta$, $d = H\eta$, where
\begin{subequations}\label{Eq:InternalModel}
\begin{align}\label{Eq:InternalModel-F}
F &= \mathrm{blkdiag}\left(\left[\begin{smallmatrix}
0
\end{smallmatrix}\right],\left[\begin{smallmatrix}
0 & 1\\
-\omega_1^2 & 0
\end{smallmatrix}\right],\left[\begin{smallmatrix}
0 & 1\\
-\omega_2^2 & 0
\end{smallmatrix}\right]\right)\\
H &= \left[\begin{smallmatrix}2,\omega_1,0,\omega_2,0\end{smallmatrix}\right].
\end{align}
\end{subequations}
 With an appropriate initial condition, \eqref{Eq:InternalModel} produces the previously described disturbance $d(t)$.} Figure \ref{Fig:Sim2} plots the true disturbance and the estimated disturbance $\hat{d}(t) = H\hat{\eta}(t)$ produced by three of the proposed designs (the responses for the other three are nearly identical), {\tb along with a baseline linear-quadratic estimator (LQE) applied to the composite system \eqref{Eq:EstErrorDynamics}.}

\begin{figure}[t]
\begin{center}
\includegraphics[width=1\columnwidth]{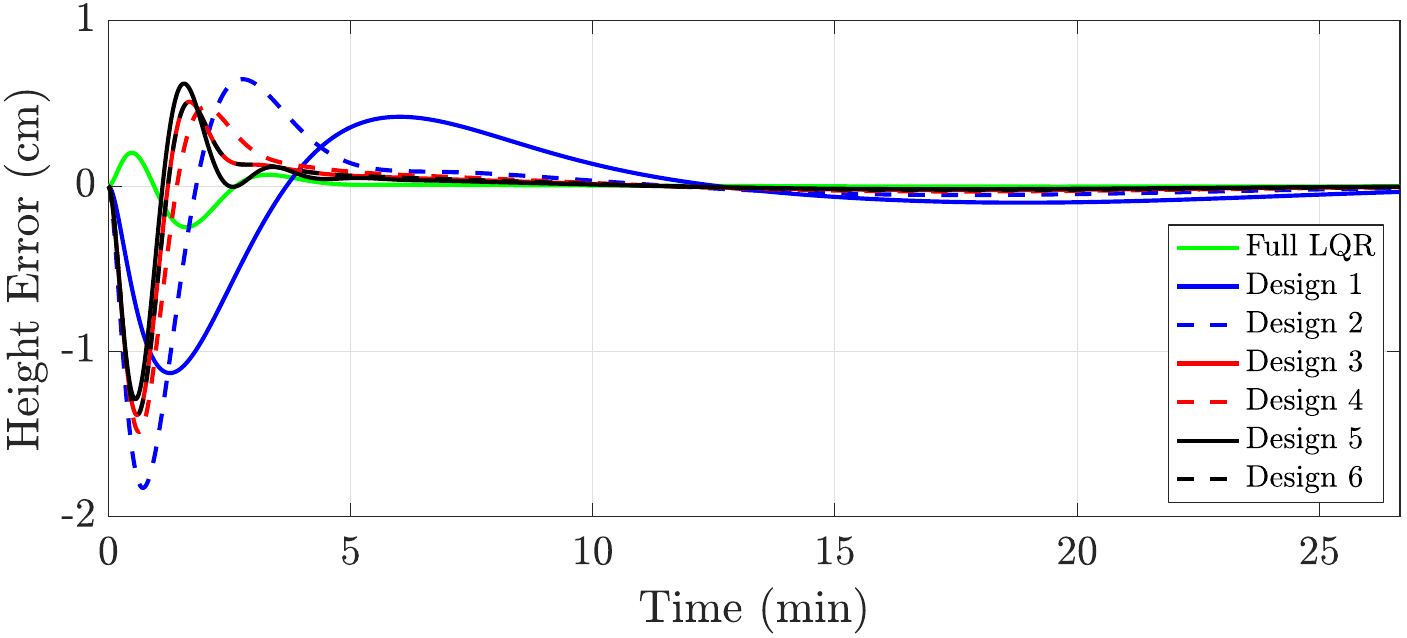}
\caption{Disturbance rejection response in four-tank system. Tuning parameters: gain $K$ computed via LQR with $Q = \mathrm{diag}(3,3,1,1)$ and $R = 0.1I$. Subsequent gains required in the designs computed via LQR with $Q = I$ and $R = \{10^6, 5\cdot 10^6, 10^4, 3\cdot10^6, 10^6,10^4\} \times I$.}
\label{Fig:Sim1}
\end{center}
\end{figure}

\begin{figure}[t]
\begin{center}
\includegraphics[width=1\columnwidth]{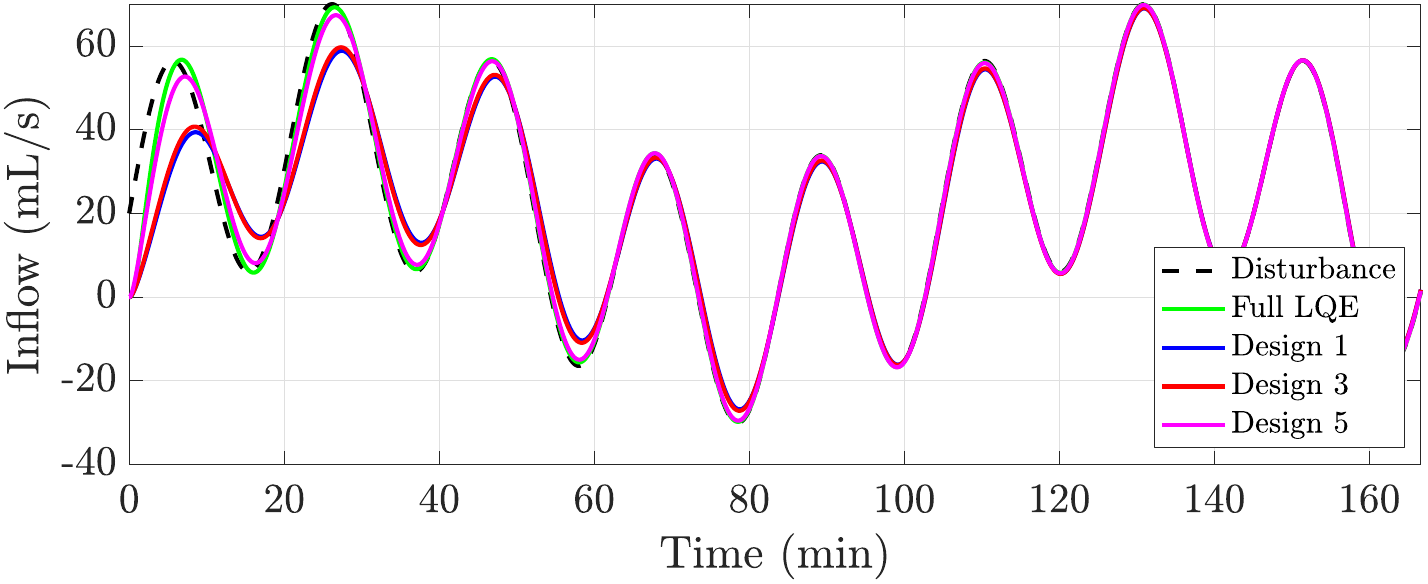}
\caption{Disturbance estimation in four-tank system. Tuning parameters: $L_x$ computed via LQR with $Q = \mathrm{diag}(3,3,1,1)$ and $R = I$. Subsequent gains required in the designs computed via LQR with $Q = I$ and $R = \{10^2, 10^2, 10^5\} \times I$.}
\label{Fig:Sim2}
\end{center}
\end{figure}
\color{black}

\section{Preliminary Results and Outlook for Nonlinear Systems}
\label{Sec:Nonlinear}

{\tb We now provide an outlook on how the main ideas of Section \ref{Sec:SSC} and Section \ref{Sec:Applications} may generalize in the nonlinear context.} The intent is not to comprehensively iron out all the details of such a development, but to lay the groundwork for a nonlinear theory paralleling the more complete linear theory of Sections \ref{Sec:SSC} -- \ref{Sec:Applications}. Section \ref{Sec:NonlinearInvariance} defines nonlinear SSC operators for nonlinear cascades in which one subsystem is LTI. Section \ref{Sec:StabilizationNonlinear} identifies which of our linear stabilization results have known generalizations to nonlinear cascades, thereby identifying unexplored directions for nonlinear design. Section \ref{Sec:ObsvCascadeNonlinear} presents a novel cascade observer design, illustrating how our catalog of linear designs can inspire new nonlinear designs. Finally, Section~\ref{Sec:MORNonlinear} identifies unexplored directions for nonlinear model reduction.

\subsection{Nonlinear Steady-State Cascade Operators}
\label{Sec:NonlinearInvariance}

Paralleling Section~\ref{Sec:SSC_DI}, we now study two nonlinear extensions of cascade systems, leading to corresponding invariance equations and SSC operators. In particular, 
using the same dimensions for all variables, consider the linear and nonlinear systems
\begin{equation}\label{Eq:SigmaNLPrimeSigmaL}
\Sigma_{\rm L}:\,\,\begin{cases}\,\begin{aligned}
\dot{x} &= Ax+ Bu\\
y &= Cx + Du
\end{aligned}\end{cases} \quad \Sigma_{\rm NL}':\,\,\begin{cases}\,\begin{aligned}
\dot{\eta} &= f(\eta) + g(\eta)v\\
z &= h(\eta)  
\end{aligned}\end{cases}
\end{equation}
and the cascade
$\Sigma_{\rm NL}' \to \Sigma_{\rm L}$
defined by the interconnection $u = z$ and $v = 0$; this is analogous to the cascade of Figure \ref{Fig:Cascade1}. 
We associate with \eqref{Eq:SigmaNLPrimeSigmaL} a \textit{primal invariance operator}
\begin{equation}\label{Eq:PrimalNonlinearSylvester}
\sylp(\pi)(\eta) \define \dfrac{\partial \pi}{\partial \eta}(\eta)f(\eta) - A \pi(\eta),
\end{equation}
where the notation $\sylp(\pi)(\eta)$ indicates the operator $\sylp$ acting on a function $\pi$ {\tb as its argument}, with the result being evaluated at $\eta \in \real^{\nu}$. Assuming (see Remark \ref{Rem:SolveInvariance}) that the abstract linear differential operator $\sylp$ is invertible, the associated invariance equation
\begin{equation}
\label{Eq:PrimalInvarianceEquation}
\sylp(\pi)(\eta) = \dfrac{\partial \pi}{\partial \eta}(\eta)f(\eta) - A \pi(\eta) = B h(\eta)
\end{equation}
will possess a unique solution $\pi = \sylp^{-1}(Bh)$. We now choose to think of $(A, B, C, D)$ and $f$ as fixed data and interpret the solution as a function of the data $h$. Based on this, we call 
\begin{equation}\label{Eq:NonlinearPrimalSSC}
\begin{aligned}
\lp(h)(\eta) &\define
C\pi(\eta) + Dh(\eta)\\
&= C\sylp^{-1}(Bh)(\eta) + Dh(\eta)
\end{aligned}
\end{equation}
the \textit{primal nonlinear SCC operator}. Equation \eqref{Eq:PrimalInvarianceEquation} generalizes the Sylvester equation $\Sylp(\Pi) = \Pi F - A\Pi = BH$ from \eqref{Eq:GenSyl1} and endows the cascade $\Sigma_{\rm NL}' \to \Sigma_{\rm L}$ with an invariant manifold $\setdef{(x,\eta)}{x = \pi(\eta)}$. Similarly, the nonlinear SSC operator $\lp$ generalizes the linear matrix SSC operator $\Lp$ from \eqref{Eq:GenL1}. If $A$ is Hurwitz, the trajectories of the above cascade converge to the invariant manifold, on which the driving system evolves as $\dot{\eta} = f(\eta)$. Thus, we obtain the unforced dynamics on the
invariant manifold, which are now simply (cf. \eqref{Eq:CascadePrimalObservation})
\[
\dot \eta = f(\eta) \qquad
y = \lp(h(\eta)).
\]

Theorem \ref{Thm:SSC2} now has immediate implications for the linearization of these dynamics. Observe that with $F \define \frac{\partial f}{\partial \eta}(0)$ and $H \define \frac{\partial h}{\partial \eta}(0)$, the linearization of $\pi(\eta)$ at the origin is given by $\frac{\partial \pi}{\partial \eta}(0) = \Sylp^{-1}(BH)$, and correspondingly, the linearization of the SSC operator $\lp(h)(\eta)$ is $\frac{\partial \lp(h)}{\partial \eta}(0) = \Lp(H)$. All results of Theorem \ref{Thm:SSC2} may now be applied as needed; for example, the linearized dynamics on the invariant manifold will be detectable if $(F,H)$ is detectable and if $\Ros_{\Sigma}(\lambda)$ has full column rank for all $\lambda \in \eig(F) \cap \complex_{\geq 0}$. Beyond linear analysis, one might ask under what conditions the pair $(f,\lp(h))$ possesses a corresponding \emph{nonlinear} detectability property. In the spirit of Theorem \ref{Thm:SSC2}, one would perhaps impose such a detectability property on the pair $(f,h)$, then seek to impose an appropriate \emph{nonlinear} non-resonance condition (e.g., \cite{LM-AI-AS:04,LW-LM-CW-HS:20}) between $\Sigma = (A,B,C,D)$ and the driving dynamics $\dot{\eta} = f(\eta)$. While such a development is outside our present scope, in Section \ref{Sec:ObsvCascadeNonlinear} we will place assumptions on $\ld(h)$ that are sufficient for successful observer design in the framework of metric-based differential dissipativity \cite[Section 4]{PB-VA-DA:22}.

As a second case of interest, consider the systems
\begin{equation}\label{Eq:SigmaNLSigmaLPrime}
\Sigma_{\rm NL}:\,\,\begin{cases}\,\begin{aligned}
\dot{x} &= a(x) + b(x)u\\
y &= c(x) + d(x)u
\end{aligned}\end{cases} \quad \Sigma_{\rm L}':\,\,\begin{cases}\,\begin{aligned}
\dot{\eta} &= F\eta + Gv\\
z &= H\eta + Jv
\end{aligned}\end{cases}
\end{equation}
and the cascade $\Sigma_{\rm NL} \to \Sigma_{\rm L}^{\prime}$ defined by the interconnection $v = y$, akin to the cascade of Figure \ref{Fig:Cascade2}. We associate with this cascade the \emph{dual invariance operator}
\begin{equation}\label{Eq:DualNonlinearSylvester}
\syld(\mu)(x) \define \tfrac{\partial \mu}{\partial x}(x)a(x) - F\mu(x)
\end{equation}
and an associated invariance equation
\begin{equation}\label{Eq:DualInvarianceEquation}
\syld(\mu)(x) = \tfrac{\partial \mu}{\partial x}(x)a(x) - F\mu(x) = Gc(x).
\end{equation}
Assuming again invertibility of this operator {\tb (see Remark \ref{Rem:SolveInvariance})}, we express the solution of \eqref{Eq:DualInvarianceEquation} as $\mu = \syld^{-1}(Gc)$ which is a function of the data $G$, and define the {\tb \emph{dual nonlinear SSC operator}}
\begin{equation}\label{Eq:NonlinearDualSSC}
\begin{aligned}
\ld(G)(x) &\define -\tfrac{\partial \mu}{\partial x}(x)b(x) + G d(x)\\
&= -\left[\tfrac{\partial}{\partial x}\syld^{-1}(Gc)(x)\right]b(x) + Gd(x).
\end{aligned}
\end{equation}
The equation \eqref{Eq:DualInvarianceEquation} generalizes the Sylvester equation $\Syld(M) = MA - FM = GC$ from \eqref{Eq:GenSyl2}, and endows the cascade $\Sigma_{\rm NL} \to \Sigma_{\rm L}^{\prime}$ with an invariant manifold $\setdef{(x,\eta)}{\eta = \mu(x)}$. Similarly, the nonlinear SSC operator \eqref{Eq:NonlinearDualSSC} generalizes the linear matrix operator $\Ld(G)$ from \eqref{Eq:GenL2}. 

With error variable $\zeta = \eta - \mu(x)$ capturing the distance to the invariant manifold, routine calculations now show that (cf. \eqref{Eq:CascadeDualTransformed})
\[
\begin{aligned}
\dot{\zeta} = F\zeta + \ld(G)(x)u.
\end{aligned}
\]
Similar comments to before hold regarding relations between the linearization of these dynamics and the SSC operator results of Theorem \ref{Thm:SSC}; for instance, the controllability of the linearization was leveraged in \cite[Lemma 1]{DA-LP-LM:22b}. We discuss in the next subsections how the invariance equations \eqref{Eq:PrimalInvarianceEquation}, \eqref{Eq:DualInvarianceEquation} and nonlinear SSC operators \eqref{Eq:NonlinearPrimalSSC}, \eqref{Eq:NonlinearDualSSC}  operators can be 
exploited
in the contexts of cascade stabilization, estimation and model reduction.



\begin{remark}[{\tb Existence of Nonlinear SSC Operators}]\label{Rem:SolveInvariance}
{\tb Existence of the nonlinear SSC operators can be ensured under appropriate technical conditions. For the case of $\lp$, if}
\begin{itemize}
\item[(A0)] {\tb $f$ and $h$ are continuously differentiable with $f(0) = h(0) = 0$,}
    \item[(A1)] the {\tb (unique)} solution to $\dot{\eta} = f(\eta)$ evolves in an open forward-invariant set $\mathcal{O} \subset \real^{\nu}$ {\tb containing the origin}, and
    \item[(A2)] $A$ is Hurwitz,
\end{itemize}
then
\begin{equation}\label{eq:solutionInvariancePi}
    \pi(\eta) =   \int_{-\infty}^0 e^{-As}B h(\bar \eta(s,\eta))\,\mathrm{d} s,
\end{equation}
{\tb is the unique continuously differentiable solution to \eqref{Eq:PrimalInvarianceEquation} satisfying the boundary condition $\pi(0) = 0$, where $\bar \eta(t,\eta_0)$ denotes the solution to $\dot \eta = f(\eta)$
at time $t$ with initial condition $\eta_0$; see \cite[Theorem 2.4]{VA-LP:06} for a related result, and \cite{AI-CIB:08} for appropriate nonlinear notions of steady-state.} Note that (A1)--(A2) are sufficient to ensure that the eigenvalues of $A$ and $\tfrac{\partial f}{\partial \eta}(0)$ are disjoint, ensuring unique solvability of the Sylvester  {\tb equation arising via linearization of \eqref{Eq:PrimalInvarianceEquation} at the origin. It follows that $\lp$ defined in \eqref{Eq:NonlinearPrimalSSC} is well-defined, with domain being the set of continuously differentiable functions $h$ vanishing at the origin.} {\tb Similarly for the case of $\ld$, if}
\begin{itemize}
\item[(B0)] {\tb $a, b, c, d$ are continuously differentiable with $a(0) = c(0) = 0$}
    \item[(B1)]   the origin of $\dot x = a(x)$ is globally asymptotically stable, and
    \item[(B2)] {\tb $\eig(F) \subset \CC_{\leq 0}$ with all eigenvalues on the imaginary axis being semi-simple,}
\end{itemize}
{\tb then the unique continuously differentiable solution to the invariance equation \eqref{Eq:DualInvarianceEquation} satisfying the boundary condition $\mu(0) = 0$ is}
\begin{equation}\label{eq:solutionInvarianceMu}
    \mu(x) = \int_0^\infty e^{Fs}G c(\bar x(s,x))\,\mathrm{d}s,
\end{equation}
where $\bar x(t,x_0)$ denotes the solution to $\dot x = a(x)$ at time $t$ with initial condition $x_0$. See, for instance, 
\cite[Lemma IV.2]{FM-LP:96}. {\tb It follows that $\ld$ in \eqref{Eq:NonlinearDualSSC} is a well-defined map on matrices $G$ into the space of continuous functions.} Similar to the comments in Remark 3, (A2) and (B1) could be enforced via a preliminary state-feedback design.

    The integral solutions above are complicated to obtain in closed-form. Classical numerical solutions based on power series expansions have been proposed in, e.g., \cite{AJK:92}; see \cite[Chapter 4]{JH:04} for an extensive exposition. Neural network solutions have been proposed in, e.g., \cite{JW-JH-SSTY:00}, and more recently in  \cite{JP-MN:21} to approximate the PDE \eqref{Eq:DualInvarianceEquation}. \tb{Recent results based on the Galerkin residual method have been presented in \cite{CD-AA-DK-AM-GS-JS:24}.} It is presently an open question how to extend the frequency response methodologies
    of Section~\ref{Sec:Moments} to this nonlinear case. 
    \hfill \oprocend
\end{remark}



\subsection{Stabilization of Nonlinear Cascade Systems}
\label{Sec:StabilizationNonlinear}

The stabilization of nonlinear systems in cascade form has been extensively studied since the late 1980's; we restrict our attention to the nonlinear approaches mostly closely related to the stabilization procedures proposed in Section~\ref{Sec:Cascade-Stab}.
First, a general (often, recursive) methodology 
called ``forwarding'' was developed in \cite{FM-LP:96}, see also \cite[Chapter 6.2]{RS-MJ-PK:97}. Most available forwarding results 
focus on the cascade $\Sigma_{\rm NL} \to \Sigma_{\rm L}^\prime$ described below \eqref{Eq:SigmaNLSigmaLPrime}. 

\color{black}
When $F=0$, a summary of different design choices
can be found in \cite[Section III]{DA-LP:17}.
This covers the designs of 
Sections~\ref{Sec:Cascade-Stab-Ld}-a) and
the small-gain approach in  \ref{Sec:Cascade-Stab-Ld-LowGain}-a).
Such approaches have been also extended
to a contraction framework
in \cite{MG-DA-VA-LM:22} in order to achieve global properties.
Again for the case with $F = 0$, the low-gain design of Section~\ref{Sec:LowGainStabLp}-a)
has been extended to the nonlinear case in \cite{JWSP:20a,PL-GW:23}
via  singular perturbations techniques.
\color{black}

Motivated by output regulation problems, such designs have been further extended (under appropriate conditions) to the case where the spectrum of $F$ is simple and 
lies on the imaginary axis
\cite{DA-LP-LM:22b}, and extended in the context of contraction and incremental stability, see, e.g. \cite{MG-DA-VA-LM:24}. We further note strong similarities with the so-called ``immersion and invariance'' approach (e.g., \cite{AA-RO:03}) where the same invariance equations are used for feedback design. In this approach however, the target system is selected as a virtual stable system, and is not considered as a part of the dynamics.

To the best of our knowledge, the following problems inspired by the LTI results of Section \ref{Sec:Applications} remain open in the context of nonlinear systems:
the stabilization of cascade $\Sigma_{\rm NL}\to \Sigma_{\rm L}^\prime$ with an unstable matrix $F$; the stabilization of  cascades 
$\Sigma_{\rm L}\to \Sigma_{\rm NL}^\prime$; 
the extension of the stabilization procedures proposed in 
Sections~\ref{Sec:Cascade-Stab-Ld}-b),  \ref{Sec:Cascade-Stab-Ld-LowGain}-b)
and \ref{Sec:LowGainStabLp}
{\tb for $F\neq 0$.}
Finally, we note that these open stabilization problems are directly relevant to design approaches in feedback-based optimization, wherein the stationarity conditions of certain optimization problems appear as nonlinear elements in a cascade; see \cite{GB-JC-JIP-ED:21, GB-DLM-MHDB-SB-RS-JL-FD:24} for recent work. The further development of these nonlinear stabilization approaches therefore appears to be a promising direction to enable the design of feedback-based optimization controllers for time-varying optimization problems.



\subsection{Observation of the Cascade $\Sigma_{\rm NL}^{\prime} \to \Sigma_{\rm L}$}
\label{Sec:ObsvCascadeNonlinear}

We present now a {\tb novel} nonlinear extension of the disturbance estimator design presented in Section \ref{Sec:EstDesignLp} a). In particular, consider the
cascade  $\Sigma_{\rm NL}^{\prime} \to \Sigma_{\rm L}$  of equation 
\eqref{Eq:SigmaNLPrimeSigmaL} with $D = 0$ and $v = 0$. 
Setting $u = y + \bar u$, we obtain the composite system
\begin{equation}
    \label{eq:nonlinear_sys_observed}
\begin{aligned}
\dot{\eta} = f(\eta), \qquad \dot{x} = Ax +Bh(\eta) +  B\bar u, \qquad y = Cx,
\end{aligned}
\end{equation}
where $\bar u$ is a measurable input signal. In this setting, the $\eta$-dynamics model unmeasured disturbances affecting the plant
$x$-dynamics for which a model generator $f$ is known (differently from 
\cite[Section III.A]{WHC-JY-LG-SL:15} where a linear model is selected). In the rest of this section, we assume that assumptions (A0)--(A2) of Remark \ref{Rem:SolveInvariance} hold.

{\tb Consider a candidate} full-state estimator of the form 
\begin{equation}\label{Eq:CandidateEstimatorNonlinear_orignal_coor}
\begin{aligned}
\dot{\hat{x}} &= A\hat{x} + Bh(\hat \eta) + B\bar u - w_x, \qquad \hat{y} = C\hat{x},\\
\dot{\hat{\eta}} &= f(\hat{\eta}) - w_{\eta},\\
\end{aligned}
\end{equation}
with $w_x$ and $w_\eta$ output injection terms to be selected.
To this end, 
we consider first the following 
change of coordinates
$x\mapsto \xi : = x- \pi(\eta)$
with $\pi = \sylp^{-1}(Bh)$ from \eqref{Eq:PrimalInvarianceEquation}.
In the new coordinates, \eqref{eq:nonlinear_sys_observed} becomes
$$
\begin{aligned}
\dot \eta  = f(\eta),\qquad     \dot \xi = A\xi + B\bar u, \qquad y = C\xi + \lp(h)(\eta), 
\end{aligned}
$$
where $\lp(h)(\eta)$ is the primal SSC operator defined in \eqref{Eq:NonlinearPrimalSSC}.
Similarly, in the new coordinates
$\hat x\mapsto \hat\xi : = \hat x- \pi(\hat \eta)$,
the estimator 
\eqref{Eq:CandidateEstimatorNonlinear_orignal_coor} reads
\begin{equation}\label{Eq:CandidateEstimatorNonlinear}
\begin{aligned}
\dot{\hat{\xi}} &= A\hat{\xi} + B\bar u
-w_x +\dfrac{\partial \pi}{\partial \eta}(\hat \eta)w_\eta,
\\
\dot{\hat{\eta}} &= f(\hat{\eta}) - w_{\eta},\\
\end{aligned}\qquad \begin{aligned}
\hat{y} &= C\hat{\xi} + \lp(h)(\hat \eta)\\
\hat x &= \hat \xi + \pi(\hat \eta).
\end{aligned}
\end{equation}
Letting $\tilde{x},\tilde{\eta},\tilde{y}$ denote the usual error variables, and selecting 
\begin{equation}\label{Eq:CorrectionSelections}
w_{\eta} = K_{\eta}(\hat y-{y}), \quad w_x = \tfrac{\partial \pi}{\partial \eta}(\hat \eta)w_\eta,
\end{equation}
with gain matrix $K_{\eta}$ to be chosen, 
the error dynamics are
\begin{equation}\label{Eq:CandidateEstimatorErrorNonlinear}
\begin{aligned}
\dot{\tilde{\xi}} &= A\tilde{\xi}\\
\dot{\tilde{\eta}} &= f(\eta) -f(\eta-\tilde\eta ) -K_{\eta}\tilde y\\
\tilde{y} &= C\tilde \xi + \lp(h)(\eta) - \lp(h)(\eta -\tilde\eta).
\end{aligned}
\end{equation}
Following
\cite[Section 4]{PB-VA-DA:22}, we know that if the pair $(f,\lp(h))$ is \emph{differentially detectable}, then 
one can design an observer for the $\eta$-dynamics
if the (fictitious) output $\lp(h)(\eta)$ is available; we seek now to build on this observation. To simplify the development of this section, we impose the following assumptions:
\begin{itemize}
    \item[(O1)] There exist a matrix $L \in \RR^{p\times \nu}$, a {\tb $C^1$} function $\psi:\RR^{p}\to\RR^{p}$, and a matrix $R \in \RR^{p\times p}$, $R\succ 0$, such that
    $$
    \lp(h)(\eta) = \psi(L\eta) , \qquad R\dfrac{\partial \psi}{\partial s} (s)+\dfrac{\partial \psi}{\partial s}(s)^{\sf T} R \succeq 2I_p,
    $$
    for all $\eta \in \real^{\nu}$, $s\in \RR^{p}$.
    \item[(O2)] There exist a matrix $P\in \RR^{\nu\times\nu}$, $P \succ 0$, and positive scalars $\varrho,\varepsilon > 0$ such that
    $$
P \dfrac{\partial f}{\partial \eta} (\eta) + 
\dfrac{\partial f}{\partial \eta}(\eta) ^{\sf T} P
-  2 \varrho L^{\sf T} L  \preceq -2\varepsilon I_\nu,
\quad \text{for all}\,\,\eta\in \RR^\nu.
$$
\end{itemize}
Assumption (O1) imposes that $\lp(h)$ can be expressed as the composition of a strongly monotone function $\psi$ and a linear map $L$, while Assumption (O2) is related to the differentiable detectability of the pair $(f,L)$, see \cite[Section 4.4]{PB-VA-DA:22}. {\tb Our subsequent design will make use of the matrices $(L,R,P)$, but does not require the explicit form of the function $\psi$.} We can now present the main result of this section, which extends the design of 
 Section \ref{Sec:EstDesignLp} a)
 to a particular class of nonlinear systems. 
Nonlinear extensions of the low-gain tuning estimator designs presented in Sections \ref{Sec:LowGainEstimatorLp}--\ref{Sec:LowGainEstimatorLd} and Remark \ref{Rem:TuningEstimator} are deferred to future work.




\color{black}

\begin{theorem}[\bf Nonlinear Cascade Estimator]\label{Thm:NonlinearTuningEstimator}
  Suppose Assumptions \textnormal{(A0)--(A2)} of Remark \ref{Rem:SolveInvariance} and Assumptions \textnormal{(O1)--(O2)} hold and set $K_\eta = \kappa P^{-1} L^{\sf T} R$ with any $\kappa\geq \varrho$. Then \eqref{Eq:CandidateEstimatorNonlinear_orignal_coor} with $w_x, w_{\eta}$ as in \eqref{Eq:CorrectionSelections} is a global observer 
  for the cascade
\eqref{eq:nonlinear_sys_observed}, i.e.,
    $$
    \lim_{t\to\infty} |\hat x(t)-x(t)| + |\hat \eta(t) - \eta(t)| = 0
    $$
    for all $(x(0),\eta(0))\in \RR^{n}\times {\mathcal O}$ and
    $(\hat \xi(0),\hat \eta(0))\in \RR^n\times {\mathcal O}$.
\end{theorem}

\begin{proof}
With $w_{\eta} = K_{\eta}(y-\hat{y})$, the error dynamics
\eqref{Eq:CandidateEstimatorErrorNonlinear} read
\begin{equation}\label{eq:error_estimation_proof}
\begin{aligned}
\dot{\tilde{\xi}} &= A\tilde{\xi}, 
\\
\dot{\tilde{\eta}} &=  \tilde f(\tilde \eta,\eta)- \kappa P^{-1}L^{\sf T} R (\tilde{y}_x  + \tilde{y}_\eta),
\end{aligned}
\qquad 
\begin{aligned}
    \tilde{y}_x & = C\tilde \xi , \\
    \tilde{y}_\eta & = \tilde \psi(L\tilde \eta,L\eta),
\end{aligned}
\end{equation}
with the notation 
$\tilde f(\tilde\eta, \eta)
 = f(\eta)-f(\eta-\tilde\eta )$ and 
$\tilde \psi(L\tilde\eta, L\eta)
 = \psi(L\eta)-\psi(L(\eta-\tilde\eta ))$.
 In view of A1), the $\tilde \xi$-dynamics is 
 globally exponentially stable.
As a consequence, in view of standard ISS results
 for cascade systems (e.g. 
\cite[Chapter 10]{AI:99}), the error dynamics \eqref{eq:error_estimation_proof}
 is globally asymptotically stable if the $\tilde \eta$-dynamics is ISS with respect to $\tilde{y}_x$. Furthermore, $(\tilde \xi,\tilde \eta) = 0$ implies
 $(\tilde x, \tilde \eta)= 0$ and so the statement of the theorem.
 
 To this end,  
 consider the Lyapunov function 
 $V(\tilde{\eta}) = \tilde \eta^{\sf T} P \tilde \eta$.
 Its derivative along solutions to \eqref{eq:error_estimation_proof} yields
 \begin{align}
     \dot V & = 2\tilde \eta^{\sf T} P \big[ \tilde f(\tilde \eta,\eta)- \kappa P^{-1}L^{\sf T} R (\tilde{y}_x  + \tilde{y}_\eta) \big] \notag
     \\
      & \leq 2\tilde \eta^{\sf T} P \widetilde \Phi(\tilde \eta, \eta)
      + \varepsilon |\tilde \eta|^2 + 
      \tfrac{\kappa^2}{\varepsilon} |L^{\sf T} R \tilde{y}_x|^2, \label{eq:ineqV_proof}  
 \end{align}
with the compact notation
$\widetilde \Phi (\tilde \eta, \eta) 
     : = 
\Phi(\eta) - \Phi(\eta - \tilde \eta)$
and
$\Phi (\eta) : = f(\eta) - \kappa
P^{-1}L^{\sf T} R \psi(L\eta)$.
Using the mean-value theorem, we have 
\begin{equation}
    \label{eq:mean_value_proof}
    2\tilde \eta^{\sf T} P \widetilde\Phi(\tilde \eta,\eta)
       = 
        \tilde \eta^{\sf T} \!\int_0^1 
       \left[P \dfrac{\partial \Phi}{\partial \eta} (\eta + s\tilde\eta)  + 
       \dfrac{\partial \Phi}{\partial \eta} (\eta + s\tilde\eta)^{\sf T}  P\right] \mathrm{d}s \, \tilde \eta
\end{equation}
and if we show that the following inequality
holds 
\begin{equation}
    \label{eq:ineq_proof}
    P \dfrac{\partial \Phi}{\partial \eta} (\eta)  + 
       \dfrac{\partial \Phi}{\partial \eta} (\eta)^{\sf T}  P
       \preceq - 2\varepsilon I_\nu, 
    \quad \text{for all}\,\,\eta\in \RR^\nu.
\end{equation}
then, 
from 
\eqref{eq:mean_value_proof} we get
$
 2\tilde \eta^{\sf T} P \widetilde\Phi(\tilde \eta,\eta)\leq 
 -2\varepsilon |\tilde \eta|^2
$
and from \eqref{eq:ineqV_proof} we further obtain 
$\dot V \leq - \varepsilon|\tilde \eta|^2 + \gamma |y_x|^2$ with $\gamma = \tfrac{\kappa^2}{\varepsilon}|L^{\sf T} R|^2$,
showing the desired ISS-property of the $\tilde \eta$-dynamics with respect to $\tilde \xi$ and concluding the proof. So we are left with showing 
the inequality \eqref{eq:ineq_proof}.
To this end, using the definition of $\Phi$  we obtain 
$$
\dfrac{\partial \Phi}{\partial \eta}(\eta) = 
 \dfrac{\partial f}{\partial \eta} (\eta)
 - \kappa P^{-1}L^{\sf T} R
  \dfrac{\partial \psi}{\partial \eta} (L\eta) L.
$$
 Finally, combining inequality \eqref{eq:ineq_proof}
and properties (O1)--(O2) we get
\begin{equation*}
   \begin{aligned}
  &  P \dfrac{\partial \Phi}{\partial \eta} (\eta)  + 
       \dfrac{\partial \Phi}{\partial \eta} (\eta)^{\sf T}  P
        = 
        \\
    &      P \dfrac{\partial f}{\partial \eta} (\eta)  + 
       \dfrac{\partial f}{\partial \eta} (\eta)^{\sf T}  P
      - \kappa L^{\sf T} \left[  R\dfrac{\partial \psi}{\partial \eta} (L\eta)  
      + \dfrac{\partial \psi}{\partial \eta}(L\eta)^{\sf T}   R \right] L
       \preceq
\\
 &\qquad \quad P \dfrac{\partial f}{\partial \eta} (\eta)  + 
       \dfrac{\partial f}{\partial \eta} (\eta)^{\sf T}  P
      -  2\varrho L^{\sf T}  L - 2(\kappa - \varrho) L^{\sf T} L \preceq -2 \varepsilon I_\nu
   \end{aligned}
\end{equation*}
for all $\kappa\geq\varrho$ and all $\eta\in \RR^\nu$, concluding the proof.
\end{proof}

\smallskip

{\tb
\begin{remark}[\bf Estimator Structure]
For the design of state estimators having the same state dimension as the plant, differential detectability is generally required to ensure convergence, as shown in \cite{PB-VA-DA:22}. Conditions (O1)--(O2) specifically impose a differential detectability property with respect to a Euclidean metric. Consequently, these conditions could be relaxed by adopting a more general Riemannian framework  \cite[Chapter 4]{PB-VA-DA:22}. Alternative observability assumptions could be made, at the cost of employing a more complex class of observers (e.g., KKL observers) for the $\eta$-dynamics. \hfill \oprocend 
\end{remark}
}

\smallskip

\subsubsection*{An example} We illustrate the observer design procedure of Theorem \ref{Thm:NonlinearTuningEstimator} with the following academic example. The plant $\Sigma_{\rm L} = (A,B,C,0)$ in \eqref{eq:nonlinear_sys_observed} with state $x = (x_1,x_2) \in \real^2$ is given by
\[
A = \begin{bmatrix}    -\alpha_1 & \alpha_2\\  -\alpha_3 & -\alpha_4\end{bmatrix}, \quad B = \begin{bmatrix}
    0 & \beta_2 \\  \beta_3 & \beta_4
\end{bmatrix},\quad C = \begin{pmatrix}
    \kappa_1 & \kappa_2
\end{pmatrix},
\]
with $\alpha_1, \alpha_3, \alpha_4 > 0$ and  $\alpha_2 = \alpha_1$. The disturbance dynamics $\Sigma_{\rm NL}^{\prime}$ with state $\eta = (\eta_1,\eta_2) \in \real^2$ are modelled using a Van der Pol oscillator
\[
\begin{aligned}
f(\eta) & = 
\begin{bmatrix}
    \eta_2 
    \\
    \mu(1-\eta_1^2)\eta_2 - \eta_1
\end{bmatrix},
\qquad
h(\eta) = \begin{bmatrix}
  \eta_1 +  \eta_1^3
    \\
   \eta_2
\end{bmatrix}, 
\end{aligned}
\]
with $\mu > 0$. With $\bar\alpha = \alpha_1\alpha_4 +\alpha_2\alpha_3$ and $\beta_3$ left arbitrary, the remaining parameters are selected as $a_1 = \beta_3\alpha_2\bar \alpha^{-1}$, $a_2 = \beta_3\alpha_1 \bar\alpha^{-1}$, $b_1 = 3a_1 \mu^{-1}$, 
$b_2 = 3a_2\mu^{-1}$, $c_1 =\bar\alpha^{-1} (\alpha_4 b_1+\alpha_2(b_1+\beta_3))$,
$c_2 =\bar\alpha^{-1} (-\alpha_3 b_1+\alpha_1(b_1+\beta_3))$, $\kappa_1 = b_2$, 
$\kappa_2 = -b_1$,
 $\beta_2 = c_1 + (\mu+\alpha_1)b_1-\alpha_2b_2$, 
 and $
\beta_4 = c_2 + (\mu+\alpha_4)b_2   +\alpha_3b_1$. It can then be verified that the invariance equation \eqref{Eq:PrimalInvarianceEquation}
admits a unique solution $\pi :\RR^2\to\RR^2$ with components
$$
\pi_1(\eta)  =   a_1\eta_1^3 +b_1\eta_2 +c_1 \eta_1,
\quad
\pi_2(\eta)  = a_2 \eta_1^3 +b_2\eta_2 +c_2 \eta_1,
$$
and thus $\lp(h)(\eta) = C\pi(\eta) = k \eta_{1}$ with $k = 9\beta_3^2/(\mu^2\alpha_1(\alpha_3+\alpha_4)^2) > 0$. Assumption (O1) is therefore verified with 
$\psi(s) = ks$,
$L=[1 \; 0]$, and $R=k^{-1}$. Similarly, one can verify (O2) on any compact set\footnote{To be precise, one can show the assumptions is verified on sets of the form $\{\eta\in\RR^2: |\eta_1|^2>1-\tfrac{\varepsilon}{\mu}, |\eta|\leq R\}$, $R>0$, but this is not an issue because the limit cycle of the Van der Pol oscillator is attractive from the interior, that is, when $\eta_1$ is small.}  with a $P$ of the form 
$P = \begin{psmallmatrix}
    1 & -q \\ -q & 1
\end{psmallmatrix}$ and $q>0$ small enough and $\varrho>0$ large enough. The trace of the observer estimation errors $x - \hat{x}$ and $\eta - \hat{\eta}$ from a randomized initial condition is plotted in Figure \ref{Fig:Sim3} 
with the parameters chosen as 
 $\mu = 3$, $\alpha_1 = \alpha_3 = \alpha_4 = \beta_3 = 1$, $q = 0.01$, $\kappa = 100$.

\begin{figure}[ht!]
\begin{center}
\includegraphics[width=1\columnwidth]{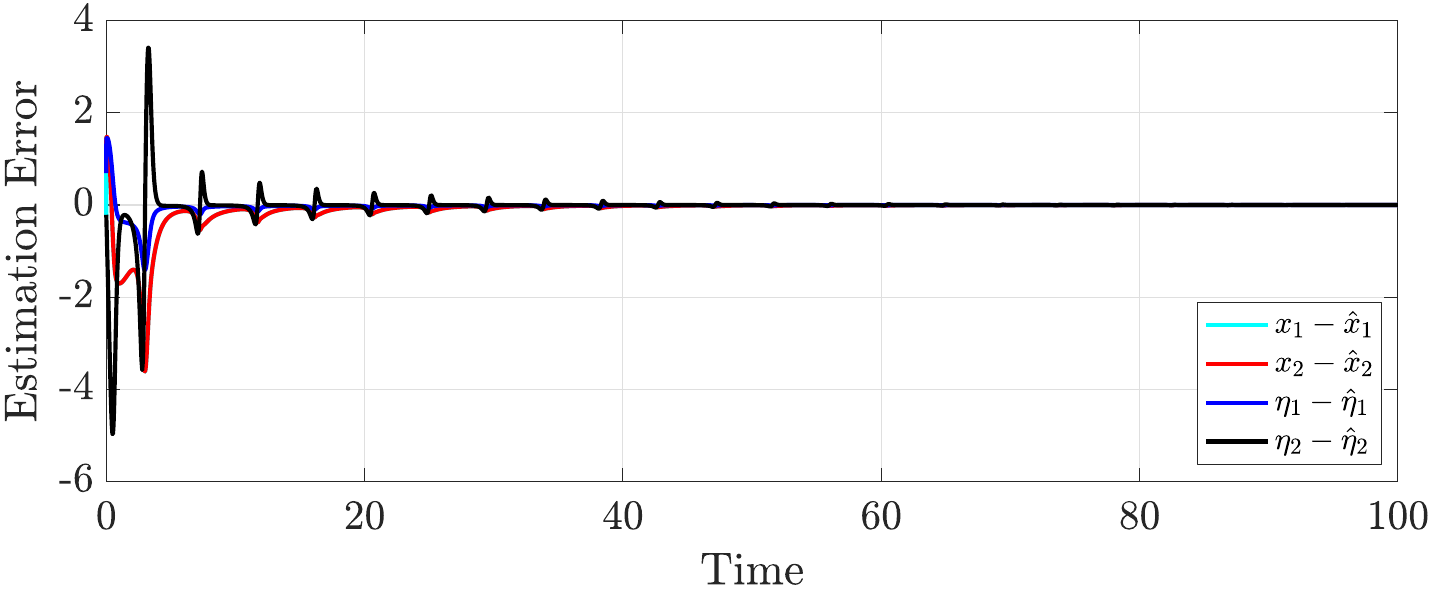}
\caption{Disturbance estimation for a linear system driven by a Van der Pol oscillator. 
}
\label{Fig:Sim3}
\end{center}
\end{figure}



\subsection{Model Order Reduction}
\label{Sec:MORNonlinear}
 
As noted in Section~\ref{Sec:Cascade-MOR}, moment matching is equivalent to matching the steady-state response of the cascade $\Sigma^{\prime} \rightarrow \Sigma$, or to matching the impulse response matrix of the cascade $\Sigma \rightarrow \Sigma^{\prime}$. This viewpoint has enabled the extension of the moment matching theory beyond linear systems: while nonlinear systems do not have a well-defined transfer function (and so classical moments), they may have well-defined interconnection responses. Thus, constructing a reduced-order model that has the same steady-state output as the original system for the same class of inputs (or, the same filtered impulse response for the same filter) has become a proxy for nonlinear moment matching that is equivalent to the classical moment matching when the systems are linear. The nonlinear enhancement of the moment $\Lp(H)$ has been introduced in \cite{AA:10} for $\Sigma_{\rm NL}^{\prime} \rightarrow \Sigma_{\rm NL}$ and $\Sigma_{\rm L}^{\prime} \rightarrow \Sigma_{\rm NL}$, while the nonlinear enhancement of the moment $\Ld(G)$ has been given in \cite{TI-AA:15} for $\Sigma_{\rm NL} \rightarrow \Sigma_{\rm NL}^{\prime}$ and $\Sigma_{\rm NL}^{\prime} \rightarrow \Sigma_{\rm NL} \rightarrow \Sigma_{\rm NL}^{\prime}$. A nonlinear Loewner framework has been presented in \cite{JDS-AA:21}.
Similar results have been provided for very general classes of systems, such as systems with time delays \cite{GS-AA:16}. A survey of the resulting ``interconnection-based'' model order reduction theory is given in  \cite{GS-AA:24}. The characterization of the structural properties of families of nonlinear reduced-order models based on nonlinear non-resonance-type conditions is an open question and a natural extension of the results in Section~\ref{Sec:Cascade-MOR}.



\section{Conclusions}
\label{Sec:Conclusions}


This work provided a unified study of steady-state cascade operators, our terminology to indicate the linear operators which naturally appear in problems of moment-based model-order reduction, and in stabilization and estimation problems involving cascaded linear systems (Section \ref{Sec:SSC_DI}). We have characterized system-theoretic properties of the operators (Section \ref{Sec:SystemTheoreticSSC}), their relation to frequency response and moments (Section \ref{Sec:Moments}), and catalogued analysis and design methodologies for the above application areas which directly leverage distinct properties of the SSC operators (Section \ref{Sec:Applications}). Notably, even in the LTI case, many of the presented design pathways are novel, particularly for the case of cascade estimation. In Section \ref{Sec:Nonlinear} we sketched a nonlinear theory of SSC operators, and provided evidence that the linear theory of Section \ref{Sec:SSC} can indeed inspire new design methodologies based on nonlinear SSC operators. Finally, we remark that some of the results herein
can be (or have already been) extended to {\tb discrete-time  (e.g., \cite{DB-AM-AA:24}) or infinite-dimensional systems (e.g., output regulation
\cite{LP-SP-TH:08,LP:16,NV-LB:23}, repetitive control  \cite{DA-SM-NvdW:21}, stabilization
\cite{VN:21,SM-LB-DA:21,SM-DA-VA:22,NV-LB:23}, model reduction \cite{TCI-OVI:12,GS-AA:16}).}

{\tb One direction for future study is the extension of the theoretical and design results of Sections  \ref{Sec:SSC} and \ref{Sec:Applications} to recursive design for more general cascaded interconnections, such as so-called (strict) feedforward systems \cite{FM-LP:96}. In this context, effort is needed to understand whether the properties of the resulting SSC operators (as in Theorem \ref{Thm:SSC}) can be assessed based on open-loop plant data, and how the low-gain design procedures of Section \ref{Sec:Applications} can be modified to prevent closed-loop performance degradation when they are repeatedly applied.}

{\tb A major open direction is the development of a similarly comprehensive set of analysis and design results for SSC operators for time-varying and nonlinear systems, mirroring Sections  \ref{Sec:SystemTheoreticSSC}, \ref{Sec:Moments} and \ref{Sec:Applications}.} For example, Theorems~\ref{Thm:SSC} and \ref{Thm:SSC2} suggest the introduction of appropriate nonlinear ``non-resonance'' conditions \cite{marconi2004non} may imply invertibility and (e.g., differential) stabilizability/detectability properties of nonlinear SSC operators. Similarly, Section~\ref{Sec:Moments} suggests that relationships between SSC operators and harmonic response may be obtainable under additional assumptions (cf. frequency response functions for convergent systems \cite{AP-NvdW-HN:07}), which would lead to novel low-gain stabilizer/estimator designs akin to those in Section~\ref{Sec:Applications}. {\tb However, significant challenges remain in extending these results to nonlinear systems. Although the nonlinear counterparts of \eqref{Eq:Sylvester} and \eqref{Eq:GenL} are reasonably well understood, their properties have not been thoroughly investigated. Even basic results on existence and uniqueness are scattered across the literature, often derived under varying assumptions and constraints. Developing nonlinear analogues of Theorems~\ref{Thm:SSC} and \ref{Thm:SSC2}, along with a unified theoretical framework comparable to the one presented here, remains a major challenge—one that we aim to address in future work.}

\section*{Acknowledgements}
J. W. Simpson-Porco acknowledges helpful discussions with L. Chen regarding the derivation in Theorem \ref{Thm:Moments}  leading to \eqref{Eq:LfFreqb}.

\appendices

\section{Solvability of Hautus and Dual Hautus Equations}
\label{Sec:Hautus}

This appendix contains results concerning solvability of certain linear matrix equations; the treatment here is inspired by \cite[Theorem 9.6]{HLT-AS-MH:01}. For $k, n_1, n_2, \nu \in \integer_{\geq 1}$, a matrix $F \in \complex^{\nu \times \nu}$, a set of matrices $(R_i)_{i=1}^{k}$ in $\complex^{n_1 \times n_2}$, and a set of polynomials $(q_i)_{i=1}^{k}$ with real coefficients, we define the ``primal'' Hautus operator
\[
\map{\Hp}{\complex^{n_2 \times \nu}}{\complex^{n_1 \times \nu}}, \quad \Hp(X) \define \sum_{i=1}^{k}\nolimits R_i X q_i(F),
\]
where $q_i(F)$ denotes formal substitution of $F$ as the indeterminate into the polynomial. The result below characterizes injectivity/surjectivity of this operator and provides analogous results for a ``dual'' operator
\[
\map{\Hd}{\complex^{\nu \times n_1}}{\complex^{\nu \times n_2}}, \quad \Hd(Y) \define \sum_{i=1}^{k}\nolimits q_i(F) Y R_i.
\]

\begin{theorem}[\bf Solvability of Hautus Equations]\label{Thm:HautusExtended}
Associated with the Hautus operators defined above, define the $n_1 \times n_2$ polynomial matrix
\[
R(\lambda) := \sum_{i=1}^{k}\nolimits R_iq_i(\lambda), \qquad \lambda \in \complex.
\]
Then
\begin{enumerate}[(i)]
\item \label{Itm:HautusP} $\Hp$ is surjective (resp. injective) if and only if $R(\lambda)$ has full row rank (resp. full column rank) for all $\lambda \in \eig(F)$;
\item \label{Itm:HautusD} $\Hd$ is surjective (resp. injective) if and only if $R(\lambda)$ has full column rank (resp. full row rank) for all $\lambda \in \eig(F)$.
\end{enumerate}
\end{theorem}

\begin{proof}
(i): The surjectivity statement is precisely the result of \cite[Theorem 9.6]{HLT-AS-MH:01}. To show injectivity, we endow the domain and codomain of $\Hp$ with the inner product $\langle Z_1,Z_2\rangle = \mathrm{Tr}(Z_1^{*}Z_2)$ and compute the adjoint operator $\Hp^*$ of $\Hp$, which is the unique linear operator satisfying $\langle \overbar{X}, \Hp(X)\rangle = \langle \Hp^*(\overbar{X}),X\rangle$ for all $X \in \complex^{n_2 \times \nu}$ and all $\overbar{X} \in \complex^{n_1 \times \nu}$. Since both the domain and codomain are finite-dimensional,  $\Hp$ is injective if and only if $\Hp^*$ is surjective. Routine computation of the adjoint shows that
\[
\Hp^*(\overbar{X}) = \sum_{i=1}^{k}\nolimits R_i^*\overbar{X}q_i(F)^* = \sum_{i=1}^{k}\nolimits  R_i^*\overbar{X}q_i(F^*)
\]
where real-ness of the coefficients in the polynomials $(q_i)_{i=1}^{k}$ has been used. By the previous result, $\Hp^*$ is surjective if and only if $\sum_{i=1}^{k}R_i^{*}q_i(\lambda^*)$ has full row rank for all $\lambda^* \in \eig(F^*)$, or equivalently, if $\sum_{i=1}^{k}R_iq_i(\lambda^*)^*$ has full column rank for all $\lambda^{*} \in \eig(F^*)$. Since the eigenvalues of $F$ are the complex conjugates of the eigenvalues of $F^*$, this is the same as saying $\sum_{i=1}^{k}R_iq_i(\lambda)$ has full column rank for all $\lambda \in \eig(F)$, which shows the result. 

(ii): Simply taking Hermitian transposes, note that
\[
\Hd(Y)^* = \sum_{i=1}^{k}\nolimits R_i^*Y^*q_i(F^*),
\]
which can now be viewed as a linear operator $\complex^{n_1 \times \nu} \ni Y^* \mapsto \Hd(Y)^* \in \complex^{n_2 \times \nu}$ having the same form as adjoint $\Hp^*$ computed in part (i). By analogous arguments, this operator (and hence, also $\Hd$) is surjective if and only if $R(\lambda)$ has full column rank for all $\lambda \in \eig(F)$. For the injectivity statement, note that the adjoint $\map{\Hd^*}{\complex^{\nu \times n_2}}{\complex^{\nu \times n_1}}$ of $\Hd$ may be computed to be $\Hd^*(\overbar{Y}) = \sum_{i=1}^{k}q_i(F^{*})\overbar{Y}R_i^{*}$. Transposing again, we observe that 
\[
\Hd^*(\overbar{Y})^* = \sum_{i=1}^{k}\nolimits R_i \overbar{Y}^* q_i(F)
\]
has precisely the same form as $\Hp$; we may argue in the same fashion as part (i) that $\Hd^*$ is surjective, and hence $\Hd$ is injective, if and only if $R(\lambda)$ has full row rank for all $\lambda \in \eig(F)$.
\end{proof}

\section{Low-Gain Hurwitz Stability of Block Matrices} \label{App:LGHS}

Consider the block matrix
\begin{equation}\label{Eq:Block}
\mathcal{A}(\epsilon) = \begin{bmatrix}
A + N_1(\epsilon) & N_2(\epsilon)\\
N_3(\epsilon) & F(\epsilon)
\end{bmatrix}
\end{equation}
where $A$ is Hurwitz, $N_1, N_2, N_3$ are continuous matrix-valued functions of $\epsilon \geq 0$ which are $O(\epsilon)$ as $\epsilon \to 0^+$, and where $F$ is low-gain Hurwitz stable. A Lyapunov criteria for low-gain Hurwitz stability, established in \cite{LC-JWSP:23d}, is as follows. Let $\mathsf{Q}$ denote the set of continuous symmetric matrix-valued functions of $\epsilon \geq 0$ with the property that there exist constants $\epsilon_{Q}^{\star}, c_{Q} > 0$ such that  $Q(\epsilon) \succeq \epsilon \, c_{Q} \, I_n$ for all $\epsilon \in [0,\epsilon_{Q}^{\star}]$. Similarly, let $\mathsf{P}$ denote the set of continuous symmetric matrix-valued functions of $\epsilon \geq 0$ with the property that there exists $\epsilon_{P}^{\star} > 0$ such that $P(\epsilon) \succ 0$ for all $\epsilon \in (0,\epsilon_{P}^{\star}]$. Then a matrix $A(\epsilon)$ is low-gain Hurwitz stable if and only if for each $Q \in \mathsf{Q}$ there exists $\epsilon^{\star} > 0$ and $P \in \mathsf{P}$ such that $A(\epsilon)^{\sf T}P(\epsilon) + P(\epsilon)A(\epsilon) = -Q(\epsilon)$ for all $\epsilon \in (0,\epsilon^{\star})$. 
%
%
Returning to \eqref{Eq:Block}, let $P_{A} \succ 0$ be such that $A^{\sf T}P_{A} + P_{A}A \preceq -I$, and let $P_{F}(\epsilon)$ be a Lyapunov matrix certifying low-gain Hurwitz stability of $F(\epsilon)$ as described above with $Q(\epsilon) = \epsilon I$. With the composite Lyapunov candidate $\mathcal{P}(\epsilon) = \mathrm{blkdiag}(P_{A},P_{F}(\epsilon))$, routine computations show that $\mathcal{A}(\epsilon)^{\sf T}\mathcal{P}(\epsilon) + P(\epsilon)\mathcal{A}(\epsilon)$ evaluates to
\[
-\mathcal{Q}(\epsilon) \define -\begin{bmatrix}I + M_1(\epsilon) & M_2(\epsilon)\\ M_2(\epsilon)^{\sf T} & \epsilon I
\end{bmatrix}
\]
for all sufficiently small $\epsilon \geq 0$, where $M_1, M_2$ are $O(\epsilon)$ as $\epsilon \to 0^+$. Routine Schur complement arguments now establish that $\mathcal{Q} \in \mathsf{Q}$, which shows that \eqref{Eq:Block} is low-gain Hurwitz stable.

\bibliographystyle{IEEEtran}
\bibliography{brevalias, biblio, Main, JWSP, New}

\begin{IEEEbiography}[{\includegraphics[width=1in,height=1.25in,clip,keepaspectratio]{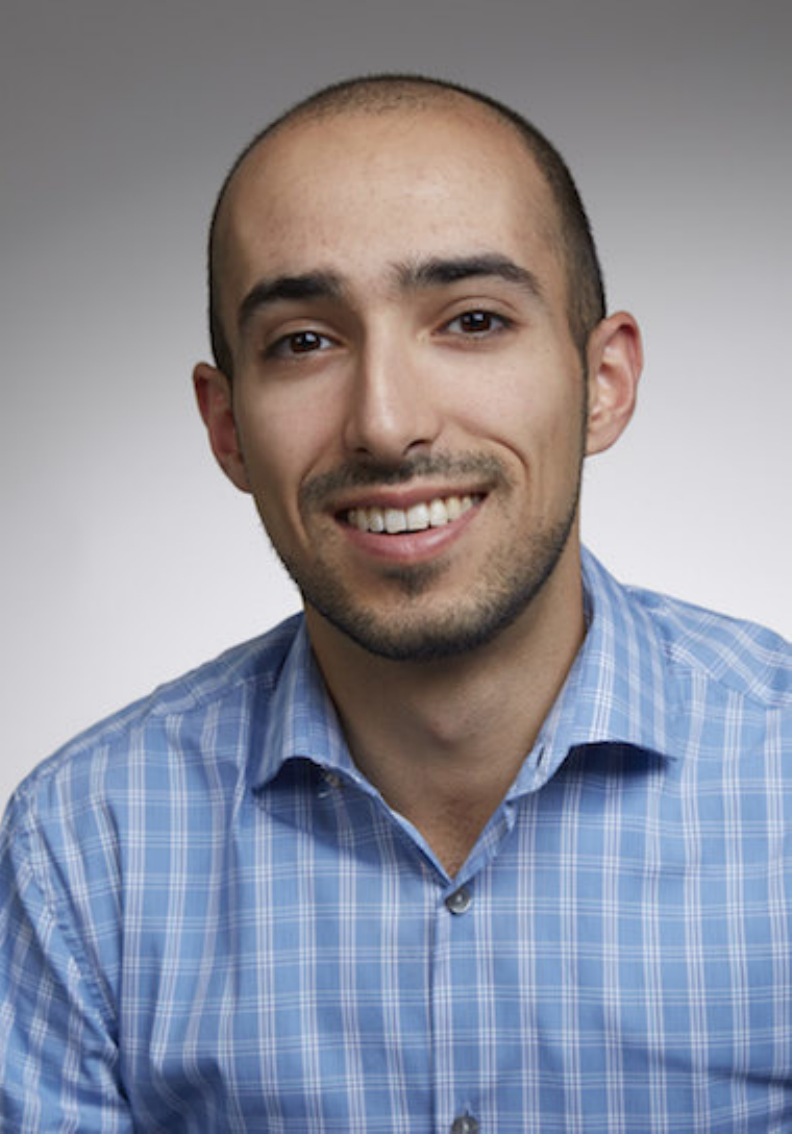}}]{John W. Simpson-Porco} (S'10--M'15--SM'22--) received the B.Sc. degree in engineering physics from Queen's University, Kingston, ON, Canada in 2010, and the Ph.D. degree in mechanical engineering from the University of California at Santa Barbara, Santa Barbara, CA, USA in 2015. He is currently an Assistant Professor of Electrical and Computer Engineering at the University of Toronto, Toronto, ON, Canada. He was previously an Assistant Professor at the University of Waterloo, Waterloo, ON, Canada and a visiting scientist with the Automatic Control Laboratory at ETH Z\"{u}rich, Z\"{u}rich, Switzerland. His research focuses on feedback control theory and applications of control in modernized power grids. He is a recipient of the Automatica Paper Prize, the Center for Control, Dynamical Systems and Computation Best Thesis Award, the IEEE PES Technical Committee Working Group Recognition Award for Outstanding Technical Report, and the Ontario Early Researcher Award.
\end{IEEEbiography}

\vspace{-2em}

\begin{IEEEbiography}
[{\includegraphics[width=1in,height=1.25in,clip,keepaspectratio]{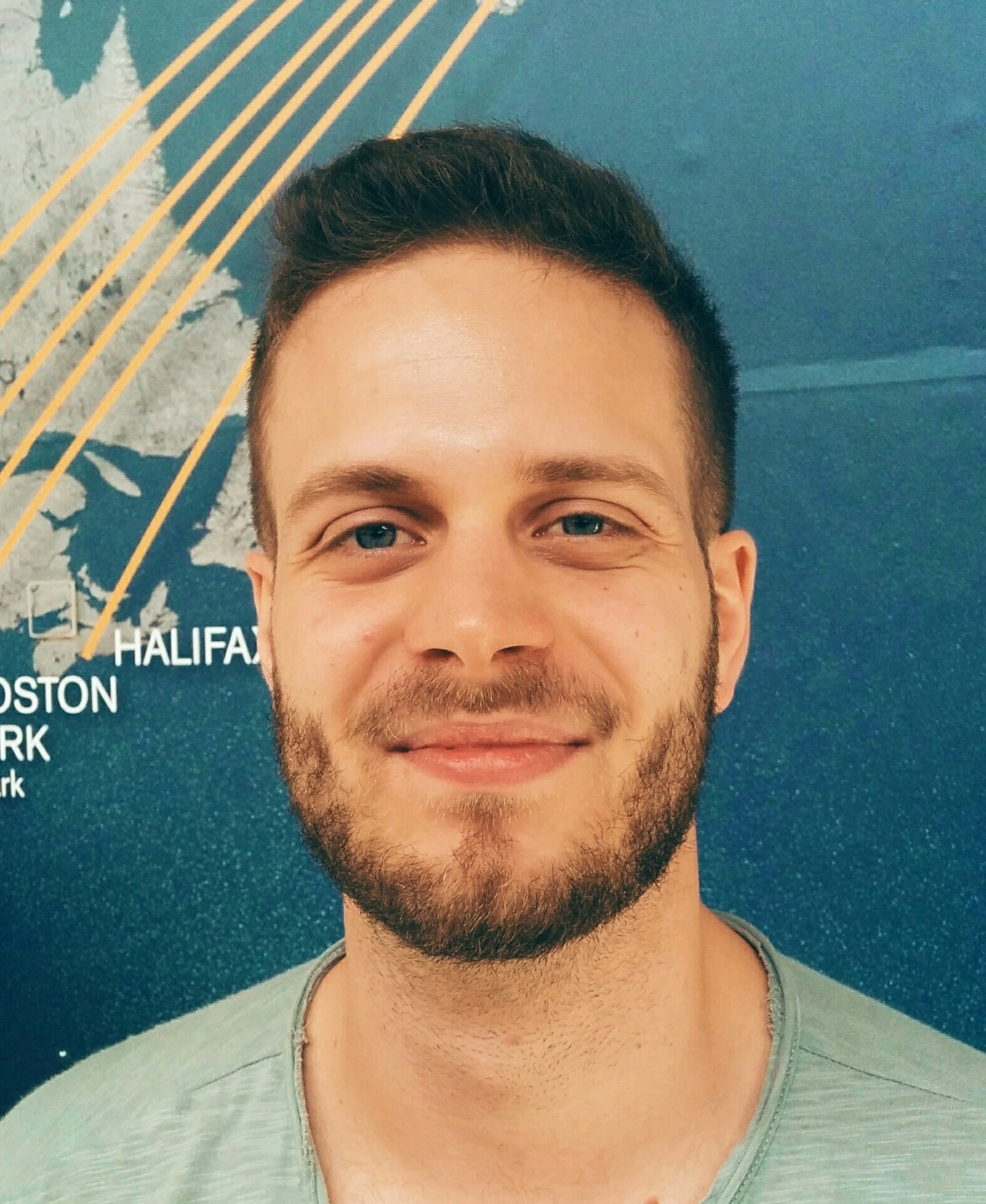}}]{Daniele Astolfi}
received the B.S. and M.S. degrees in
automation engineering from the University of Bologna,
Italy, in 2009 and 2012, respectively. He obtained a joint
Ph.D. degree in Control Theory from the University of
Bologna, Italy, and from Mines ParisTech, France, in 2016.
In 2016 and 2017, he has been a Research Assistant 
at the University of Lorraine (CRAN), Nancy, France.
Sinah ce 2018, he is a CNRS Researcher at 
LAGEPP, Lyon, France. 
His research interests include observer design, feedback
stabilization and output regulation  for nonlinear systems,
networked control systems, hybrid systems, and multi-agent systems.
 He serves as an associate 
editor of the IFAC journal Automatica since 2018 and European Journal of Control since 2023.
He was a recipient of the 2016 Best Italian Ph.D. Thesis Award in Automatica given by SIDRA.
\end{IEEEbiography}

\vspace{-2em}

\begin{IEEEbiography}
[{\includegraphics[width=1in,height=1.25in,clip,keepaspectratio]{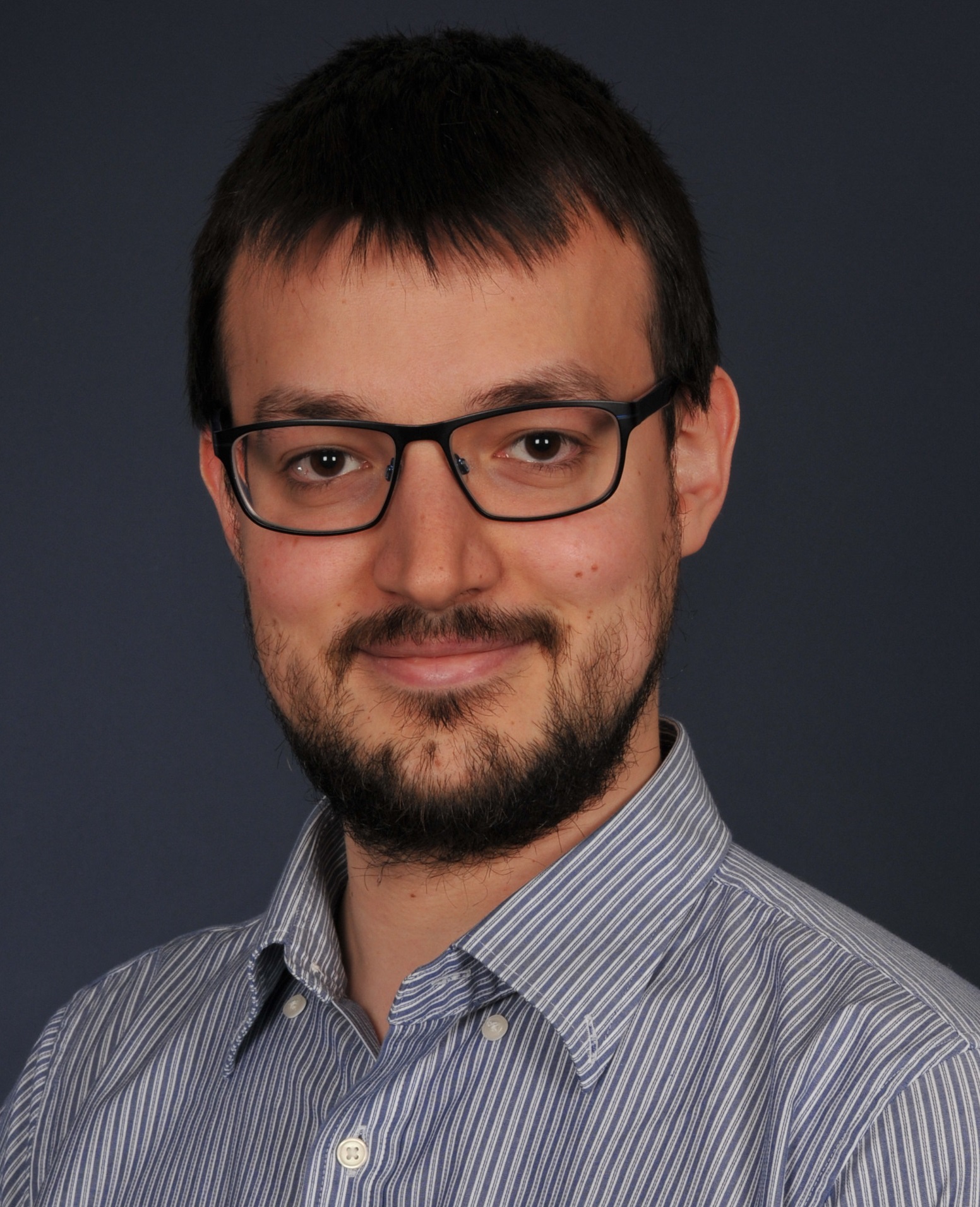}}]{Giordano Scarciotti} (Senior Member, IEEE) 
received his B.Sc. and M.Sc. degrees in Automation Engineering from the University of Rome “Tor Vergata”, Italy, in 2010 and 2012, respectively, and his Ph.D in Control Engineering and M.Sc. in Applied Mathematics from Imperial College London, UK, in 2016 and 2020, respectively. 
He is currently a Senior Lecturer (Associated Professor) at Imperial. 
He was a visiting scholar at New York University in 2015 and at University of California Santa Barbara in 2016 and a Visiting Fellow of Shanghai University in 2021-2022. He is the recipient of the IET Control \& Automation PhD Award (2016), the Eryl Cadwaladr Davies Prize (2017), an ItalyMadeMe award (2017), and the IEEE Transactions on Control Systems Technology Outstanding Paper Award (2023). He is a member of the EUCA CEB, and of the IFAC and IEEE CSS TCs on Nonlinear Control Systems. 
He is Associate Editor of Automatica. 
He was the NOC Chair for the ECC 2022 and of the 7th IFAC Conference on Analysis and Control of Nonlinear Dynamics and Chaos 2024, and the Invited Session Chair and Editor for the IFAC Symposium on Nonlinear Control Systems 2022 and 2025, respectively. 
\end{IEEEbiography}

\end{document}

%% file: fig/fig_cascade.tex
    \begin{tikzpicture}[auto, scale = 0.6, node distance=2cm,>=latex', every node/.style={scale=0.85}]
      \tikzstyle{anch} = [coordinate]
      \tikzstyle{block} = [draw, fill=white, rectangle, 
      minimum height=3em, minimum width=6em, blur shadow={shadow blur steps=5}]
      \tikzstyle{wideblock} = [draw, fill=white, rectangle, 
      minimum height=3em, minimum width=9em, blur shadow={shadow blur steps=5}]
      \tikzstyle{hold} = [draw, fill=white, rectangle, 
      minimum height=3em, minimum width=3em, blur shadow={shadow blur steps=5}]
      \tikzstyle{dzblock} = [draw, fill=white, rectangle, minimum height=3em, minimum width=4em, blur shadow={shadow blur steps=5},
      path picture = {
        \draw[thin, black] ([yshift=-0.1cm]path picture bounding box.north) -- ([yshift=0.1cm]path picture bounding box.south);
        \draw[thin, black] ([xshift=-0.1cm]path picture bounding box.east) -- ([xshift=0.1cm]path picture bounding box.west);
        \draw[very thick, black] ([xshift=-0.5cm]path picture bounding box.east) -- ([xshift=0.5cm]path picture bounding box.west);
        \draw[very thick, black] ([xshift=-0.5cm]path picture bounding box.east) -- ([xshift=-0.1cm, yshift=+0.4cm]path picture bounding box.east);
        \draw[very thick, black] ([xshift=+0.5cm]path picture bounding box.west) -- ([xshift=+0.1cm, yshift=-0.4cm]path picture bounding box.west);
      }]
      \tikzstyle{sum} = [draw, fill=white, circle, node distance=1cm, blur shadow={shadow blur steps=8}]
      \tikzstyle{input} = [coordinate]
      \tikzstyle{output} = [coordinate]
      \tikzstyle{split} = [coordinate]
      \tikzstyle{pinstyle} = [pin edge={to-,thin,black}]
      \node [input,name=input] {};
      \node [block, right of=input, node distance=2.2cm] (system) 
      {
      \makecell[c]{$
                \begin{aligned}
                \dot{\eta} &= F\eta + Gv\\
                z &= H\eta + Jv
               \end{aligned}$
                 }
     };
      \node[block, right of=system, node distance=3.5cm] (opt) {\makecell[c]{$
    \begin{aligned}
    \dot{x} &= Ax + Bu\\
    y &= Cx + Du
    \end{aligned}$}};
      \draw [thick, -latex] (input) -- node[above, name=u, pos=.2] {$v$} (system);
      \draw [thick,-latex] (system) -- node[name=z] {$u = z$} (opt);
      \node [output, right of = opt, node distance=2cm] (output) {};
      \draw [thick, -latex] (opt) -- node[name=eta, pos=0.5] {$y$} (output);

%
    \end{tikzpicture}

%% file: fig/fig_cascade2.tex
    \begin{tikzpicture}[auto, scale = 0.6, node distance=2cm,>=latex', every node/.style={scale=0.85}]
      \tikzstyle{anch} = [coordinate]
      \tikzstyle{block} = [draw, fill=white, rectangle, 
      minimum height=3em, minimum width=6em, blur shadow={shadow blur steps=5}]
      \tikzstyle{wideblock} = [draw, fill=white, rectangle, 
      minimum height=3em, minimum width=9em, blur shadow={shadow blur steps=5}]
      \tikzstyle{hold} = [draw, fill=white, rectangle, 
      minimum height=3em, minimum width=3em, blur shadow={shadow blur steps=5}]
      \tikzstyle{dzblock} = [draw, fill=white, rectangle, minimum height=3em, minimum width=4em, blur shadow={shadow blur steps=5},
      path picture = {
        \draw[thin, black] ([yshift=-0.1cm]path picture bounding box.north) -- ([yshift=0.1cm]path picture bounding box.south);
        \draw[thin, black] ([xshift=-0.1cm]path picture bounding box.east) -- ([xshift=0.1cm]path picture bounding box.west);
        \draw[very thick, black] ([xshift=-0.5cm]path picture bounding box.east) -- ([xshift=0.5cm]path picture bounding box.west);
        \draw[very thick, black] ([xshift=-0.5cm]path picture bounding box.east) -- ([xshift=-0.1cm, yshift=+0.4cm]path picture bounding box.east);
        \draw[very thick, black] ([xshift=+0.5cm]path picture bounding box.west) -- ([xshift=+0.1cm, yshift=-0.4cm]path picture bounding box.west);
      }]
      \tikzstyle{sum} = [draw, fill=white, circle, node distance=1cm, blur shadow={shadow blur steps=8}]
      \tikzstyle{input} = [coordinate]
      \tikzstyle{output} = [coordinate]
      \tikzstyle{split} = [coordinate]
      \tikzstyle{pinstyle} = [pin edge={to-,thin,black}]
      \node [input,name=input] {};
      \node [block, right of=input, node distance=2.2cm] (system) 
      {
      \makecell[c]{$
                \begin{aligned}
                \dot{x} &= Ax + Bu\\
                        y &= Cx + Du
               \end{aligned}$
                 }
     };
      \node[block, right of=system, node distance=3.5cm] (opt) {\makecell[c]{$
    \begin{aligned}
        \dot{\eta} &= F\eta + Gv\\
        z &= H\eta + Jv    
    \end{aligned}$}};
      \draw [thick, -latex] (input) -- node[above, name=u, pos=.2] {$u$} (system);
      \draw [thick,-latex] (system) -- node[name=z] {$v=y$} (opt);
      \node [output, right of = opt, node distance=2cm] (output) {};
      \draw [thick, -latex] (opt) -- node[name=eta, pos=0.5] {$z$} (output);

%
    \end{tikzpicture}